\documentclass[american,aps,pra,reprint,superscriptaddress,longbibliography]{revtex4-2}

\usepackage{amsmath,amssymb}

\usepackage{bbm}
\usepackage{color}
\usepackage{graphicx}
\usepackage[american]{babel}
\usepackage[utf8]{inputenc}
\usepackage{times}
\usepackage{epstopdf}
\usepackage{braket} 				
\usepackage{amsthm}

\newcommand{\ketbra}[2]{|#1\rangle\!\langle#2|}
\newcommand{\Tr}{\operatorname{\bf{tr}}}

\definecolor{mygrey}{gray}{0.35}
\definecolor{myblue}{rgb}{0.2,0.2,0.8}
\definecolor{myzard}{cmyk}{0,0,0.05,0}
\definecolor{mywhite}{rgb}{1,1,1}
\definecolor{myred}{rgb}{0.9,0.1,0.}
\usepackage[colorlinks=true,citecolor=myblue,linkcolor=myblue,urlcolor=myblue]{hyperref}

\newcommand{\id}{{\mathrm{id}}}
\newcommand{\diag}{\mathop{\mathrm{diag}}}
\newcommand{\vspan}{\mathop{\mathrm{span}}}

\makeatletter
\newcommand{\pushright}[1]{\ifmeasuring@#1\else\omit\hfill$\displaystyle#1$\fi\ignorespaces}
\makeatother

\newtheorem{theorem}{Theorem}
\newtheorem{lem}[theorem]{Lemma}
\newtheorem{prop}[theorem]{Proposition}

\begin{document}
	\title{Coherence of operations and interferometry}
	
	\author{Michele Masini}
	\email[]{michele.masini@ulb.be}
	\affiliation{Laboratoire d’Information Quantique, Université Libre de Bruxelles, 1050 Bruxelles, Belgium}
	\affiliation{Institute of Theoretical Physics and IQST, Universität Ulm, Albert-Einstein-Allee 11, D-89069 Ulm, Germany}
	\author{Thomas Theurer}
	\email[]{thomas.theurer@uni-ulm.de}
	\author{Martin B. Plenio}
	\email[]{martin.plenio@uni-ulm.de}
	\affiliation{Institute of Theoretical Physics and IQST, Universität Ulm, Albert-Einstein-Allee 11, D-89069 Ulm, Germany}
	
	\date{\today}
	
	\begin{abstract}
		Quantum coherence is one of the key features that fuels applications for which quantum mechanics exceeds the power of classical physics. This explains the considerable efforts that were undertaken to quantify coherence via quantum resource theories. An application of the resulting framework to concrete technological tasks is however largely missing. Here, we address this problem and connect the ability of an operation  to detect or create coherence to the performance of interferometric experiments.
	\end{abstract}
	
	\maketitle
	
	\section{Introduction\label{sec:intro}}
	The emergence of quantum technologies that outperform their classical counterparts~\cite{Feynman1982,Bennett1984,Deutsch1985,Shor1999} led to the insight that properties in which quantum mechanics departs from classical physics are not only of foundational interest, but also of practical relevance. 
	Consequently, these properties are now considered resources that can lead to operational advantages. To understand precisely which quantum property is responsible for what advantage and how to employ them optimally is the motivation for the development of various quantum resource theories~\cite{Vedral1997,Plenio2007,Horodecki2009,Aberg2006, Gour2008, Brandao2013, Horodecki2013,Baumgratz2014,Grudka2014,Veitch2014,DelRio2015,Killoran2016,Coecke2016,Theurer2017,Streltsov2017c, Tan2017, Egloff2018, Yadin2018, Chitambar2019}. Until recently, the focus was mainly on the quantification of resources present in quantum states, which is achieved with the help of static resource theories. The value of operations was determined indirectly by, e.g., their resource generation capacity~\cite{Bennett2003,Mani2015,Xi2015,Garcia2016,Bu2017}, which describes how much static resources they can create, or their resource cost~\cite{Eisert2000,Collins2001,Dana2017}, i.e., the amount of static resources that is needed to simulate them. 
	
	Only recently, the direct quantification of resources present in quantum operations started with the development of dynamical resource theories~\cite{Zhuang2018,Theurer2019,Wang2019a,Wang2019b,Liu2020,Liu2019,Gour2019a,Gour2019b,Saxena2020,Gour2020,Takagi2019a,Takagi2020a,Takagi2020b,Bauml2019,Li2020}. Compared to static resource theories, they show several advantages~\cite{Theurer2019}. 
	Firstly, quantum technologies intend to accomplish tasks that cannot be carried out with classical devices, and such devices are described by quantum operations. From a conceptual point of view, it seems thus more natural to quantify the value of operations directly without a detour through states. Secondly, dynamical resource theories are a unifying concept and include static resource theories as a special case as the preparation of a state is a specific quantum operation. Thirdly, they often allow for an operational treatment of subselection, i.e., the ability to subselect in one operation but not in another can be naturally reflected in the framework. Finally, dynamical resource theories can be used to quantify properties of quantum operations that cannot be captured by the indirect methods mentioned above. 
	One of these properties is the ability to detect coherence  in the sense that its presence makes a difference in measurement statistics~\cite{Yadin2016,Smirne2018}, which is a necessary prerequisite to its exploitation.  It is therefore equally important to investigate how well an operation can create and detect coherence, which was quantified both theoretically and experimentally in Refs.~\cite{Theurer2019,Xu2020}.	
	\begin{figure}[ht]
		\includegraphics[width=.95\linewidth]{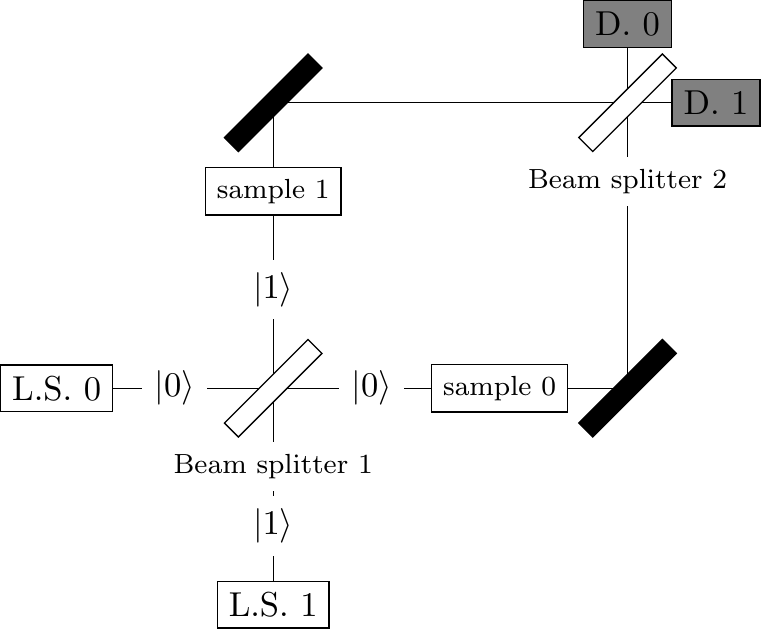}
		\caption{Mach-Zehnder interferometer. Two light sources L.S. 0 and L.S. 1 illuminate the two input ports of a beam splitter. Its two output modes experience different phases, before they reach a second beam splitter, after which they are measured by the detectors D. 0 and D. 1. }\label{fig:mach-z}
	\end{figure}

	Whilst coherence is undoubtedly central to the departure from classical physics and underlies other quantum resources such as entanglement, it is important to investigate its relevance in concrete applications~\cite{Hillery2016,Napoli2016,Matera2016}. In Ref.~\cite{Biswas2017}, the authors established a direct connection between static coherence~\cite{Winter2016,Streltsov2017c} and interferometry (see also Refs.~\cite{Bera2015,Bagan2016,Paul2017}). Here, we go one step further and apply dynamical resource theories to interferometers (see Fig.~\ref{fig:mach-z}), which allows us to give an operational meaning to both the ability to detect and to create coherence.
	
	After a quick introduction of the relevant resource theories in Sec.~\ref{subsec:rts}, we describe in Sec.~\ref{subsec:interf} the  multi-path interferometers~\cite{Biswas2017} that we are going to analyze. In Sec.~\ref{sec:main}, we then present our main results and construct families of measures that have a direct operational interpretation: they are not only proper quantifiers of resourcefulness, but also determine the advantage that dynamical coherence grants in concrete interferometric setups. 
	We then conclude in Sec.~\ref{sec:conc}.

	\section{Definitions\label{sec:def}} 
	In the following, we recall some definitions and results that will be used in the remainder of this article.

	\subsection{Resource theoretical setting}\label{subsec:rts}
	Firstly, we shortly review the dynamical resource theories that we will utilize throughout the article, but refer to the original publication~\cite{Theurer2019} for full details.
	Channel resource theories are defined by two ingredients, the set of \emph{free channels} and the set of \emph{free superchannels}.  While a quantum channel is a linear map from quantum states to quantum states, a superchannel $S$ is a linear map from quantum channels $\Theta$ to quantum channels $\Theta'$ that has a physical realization
	\begin{equation}
		\Theta'=S(\Theta):= \Psi(\Theta\otimes\id)\Phi,\label{eq:super}
	\end{equation} 
	where $\id$ is the identity channel and $\Psi$ and $\Phi$ are quantum channels~\cite{Chiribella2008}. We will use the notation $\Theta^{B\leftarrow A}$ when we need to specify that a quantum channel has an input system $A$ and an output system $B$ or simply write $\Theta^A$ when input and output system are the same.
	
	In this work, we employ one resource theory to quantify the ability of an operation to detect coherence and another one to describe its ability  to create coherence. Firstly, for every finite dimensional system  under consideration, we fix an incohererent basis, i.e., an orthonormal basis $\{\ket{i}\}_{i=0,\dots,M-1}$ and define total dephasing with respect to it as the operation $\Delta$ that acts on every state $\sigma$ as
	\begin{equation}
		\Delta(\sigma):= \sum_i\ketbra{i}{i}\sigma\ketbra{i}{i}.
	\end{equation} 
	A quantum state $\rho$ is called incoherent if it is diagonal in the incoherent basis, i.e., if $\Delta \rho=\rho$, and we denote the set of incoherent states by $\mathcal{I}$.
	
	A measurement described by a  positive operator-valued measure (POVM) with elements $P^{(n)}\ge0$, $\sum_n P^{(n)}=\mathbbm{1}$ cannot detect coherence if its outcome statistics solely depend on the populations of the states to which it is applied. This is exactly the case if it is diagonal in the incoherent basis~\cite{Theurer2019}, i.e., of the form
	\begin{equation}
		P^{(n)}=\sum_i P^{(n)}_i\ketbra{i}{i},
	\end{equation}
	where $\{\ket{i}\}$ is the incoherent basis. We will denote the set of all incoherent POVMs by $\mathcal{P}_I$. 
	To extend this definition to arbitrary instruments, we must handle subselection consistently: for a quantum instrument $\Gamma$ that allows us to apply subselection according to a variable $x$, i.e., with probability $p_x=\Tr(\Gamma_x(\rho))$, we obtain an output $\rho_x=\Gamma_x(\rho)$, we define a corresponding channel
	\begin{align*}
		\tilde{\Gamma}(\rho)=\sum_x \Gamma_x(\rho)\otimes\ketbra{x}{x}.
	\end{align*}
	Formally, we thus store the outcome $x$ in the incoherent basis of an auxiliary system, from which we can extract it at a later point using an incoherent POVM. This allows us to reduce our analysis to channels and has the additional advantage that we treat subselection in an operational manner. A channel is then unable to detect coherence if it maps incoherent POVMs to incoherent POVMs in the sense that the  populations of its output are independent of the  coherences of its input. In our resource theory that describes the detection of coherence, the set of free channels is therefore given by all operations that satisfy~\cite{Theurer2019,Liu2017,Meznaric2013}
	\begin{equation}
		\Delta\Phi=\Delta\Phi\Delta.
	\end{equation} 
	We denote this set with $\mathcal{DI}$.
	
	To describe the ability to create coherence, we choose the \emph{maximally incoherent operations} ($\mathcal{MIO}$)~\cite{Aberg2006,Liu2017,Diaz2018} as free. Treating subselection as above, these are all channels $\Psi$ that satisfy
	\begin{equation}
		\Psi\Delta=\Delta\Psi\Delta.
	\end{equation} 
	As the name suggests, this is the maximal set of channels that maps incoherent states to incoherent states, i.e., that cannot create coherence, which is why it is considered free when we investigate the creation of coherence as a resource.
	
	As free superchannels, we chose the ones with a decomposition as in Eq.~(\ref{eq:super}) where both $\Psi$ and $\Phi$ are free channels (elements of $\mathcal{MIO}$ in the creation incoherent setting or elements of $\mathcal{DI}$ in the detection incoherent setting).
	
	Thanks to the previously defined sets, we are now able to quantify the ability of an operation to detect or create coherence. In a channel resource theoretical setting, a resource measure $M(\Theta)$ is a functional from quantum channels to real numbers that satisfies the following properties. 
	\begin{itemize}
		\item \emph{Nullity}: If $\Theta$ is a free operation, then $M(\Theta)=0$.
		\item \emph{Non-negativity}: $M(\Theta)$ is non-negative. 
		\item \emph{Monotonicity}: $M(\Theta)$ is monotonic under free superoperations, i.e., $M(S(\Theta))\le M(\Theta)$ for all free superchannels $S$. This condition is equivalent to the simultaneous satisfaction of three simpler conditions, namely monotonicity under left and right composition and  monotonicity under tensor products~\cite[Prop.~12]{Theurer2019}. If we denote with $\mathfrak{F}$ the set of free channels the three conditions read
		\begin{align}
			M(\Theta)&\geq M(\Phi \Theta) \qquad\forall\Theta\quad \forall \Phi\in\mathfrak{F},\label{eq:sx}\\
			M(\Theta)&\geq M(\Theta \Phi) \qquad\forall\Theta\quad \forall \Phi\in \mathfrak{F}, \label{eq:dx}\\
			M(\Theta)&\geq M(\Theta\otimes\id^Z) \qquad\forall\Theta\quad\forall Z\label{eq:ce} .
		\end{align}
	\end{itemize}
	Moreover, one often examines two additional properties that are convenient but not necessary to consider $M(\Theta)$ a proper resource measure~\cite{Liu2019,Plenio2005}, namely
	\begin{itemize}
		\item \emph{Faithfulness}:  $M(\Theta)=0$ exactly if $\Theta$ is a free channel.
		\item \emph{Convexity}: For all $t$ with $0\leq t\leq 1$ and for all channels $\Theta_1,\Theta_2$, 
		\begin{equation}
			\qquad \quad M(t\Theta_1+(1-t)\Theta_2)\leq tM(\Theta_1)+(1-t)M(\Theta_2).
		\end{equation}
	\end{itemize}
	
	We will use the term \emph{convex measure} for functionals that  are convex and qualify as resource measures.

	\subsection{Multi-path interferometer}\label{subsec:interf}
	Next we introduce the idealized multi-path interferometer~\cite{Biswas2017} that we will investigate in this paper. To analyze the role of dynamical coherence in interferometry, we begin with an incoherent state $\tau$, as depicted in Fig.~\ref{fig:multi-path}. 
	\begin{figure}[ht]
		\includegraphics[width=.95\linewidth]{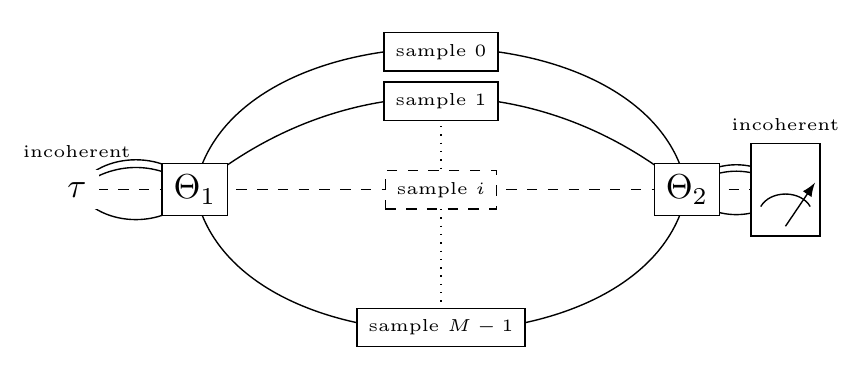}
		\caption{Sketch of a multi-path interferometer. The quantum channel $\Theta_1$ distributes an incoherent state $\tau$ to $M$ different paths that can lead to different phases via, e.g., interactions with different samples. The paths are recombined by a second channel  $\Theta_2$, and its output is incoherently measured.}\label{fig:multi-path}
	\end{figure}

	We then apply a quantum channel $\Theta_1$ to it, which distributes $\tau$ to  $M$ different paths of the interferometer represented by the orthonormal basis $\ket{0},\dots,\ket{M-1}$. Samples in the paths or different path lengths lead to relative phases  which we describe by the application of the unitary channel $\Lambda_{\vec{\phi}}$ defined as
	\begin{equation}\label{eq:UnitPhases}
		\Lambda_{\vec{\phi}}(\sigma):=\sum_{i,j}e^{i(\phi_i-\phi_j)}\ketbra{i}{i} \sigma \ketbra{j}{j},
	\end{equation}
	where $\vec{{\phi}}$ refers to the phases $\{\phi_i\}$.
	Afterwards, a channel $\Theta_2$ recombines the paths before a final incoherent measurement. The channel $\Theta_1$ is thus the generalization of beam splitter 1 in the Mach-Zehnder interferometer represented in Fig.~\ref{fig:mach-z}, $\Theta_2$ generalizes beam splitter 2, and the incoherent measurement the two detectors. 
	The goal of interferometry is now to deduce information about the relative phases from the measurement outcome. To make our setting non-trivial, we assume from here on that there exists at least one pair $k,l$ such that $\phi_k\ne \phi_l$. 
	
	Intuitively, the ability of $\Theta_1$ to generate coherence and of $\Theta_2$ to detect it will then affect how sensitive the measurement outcome is to relative phases: if $\Theta_1$ cannot create coherence, the state $\Lambda_{\vec{\phi}}\ \Theta_1 (\tau)$ is independent of $\vec{{\phi}}$. 
	On the other hand, if $\Theta_2$ cannot detect coherence as defined in Sec.~\ref{subsec:rts}, the final outcome of the incoherent measurement will be independent of the relative phases too. 
	In the following, we will make these intuitions rigorous.

	\section{Main Results}\label{sec:main}
	
	In this section, we investigate the connection between a channel's ability to detect or create coherence and its usefulness in interferometry in detail.
	To this end, we consider the setup described in Sec.~\ref{subsec:interf} and depicted in Fig.~\ref{fig:multi-path}. 
	Since the detection and creation of coherence are two different resources, we will treat them separately, beginning with the former.  
	
	\subsection{Detecting coherence and interferometry}\label{subsec:DetMain}
	To investigate which role the detection of coherence plays in interferometry, we analyze a setting that is best described in terms of a game between two parties, Alice and Bob (see Fig.~\ref{fig:op-pre}): Bob prepares a quantum state $\rho$ and sends it to Alice. 
	\begin{figure}[ht]
		\includegraphics[width=.95\linewidth]{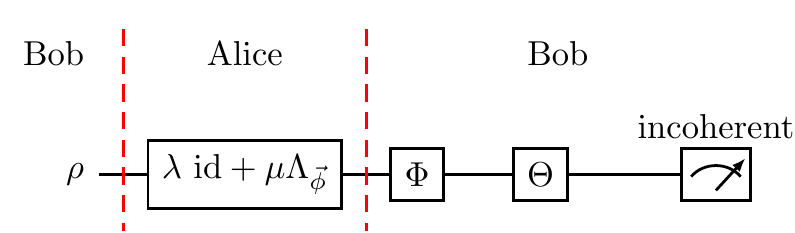}
		\caption{Schematic representation of the game played by Alice and Bob that is used to describe the role of the detection of coherence in interferometry.}\label{fig:op-pre}
	\end{figure}

	Alice applies to this state with probability $\mu$ 
	a channel $\Lambda_{\vec{\phi}}$ introduced in Eq.~\eqref{eq:UnitPhases}, and otherwise, with probability $\lambda=1-\mu$, she leaves it unchanged. She then sends the state back to Bob. His task is to guess if Alice applied $\Lambda_{\vec{\phi}}$ or not. To do this, he is allowed to first apply an arbitrary detection incoherent operation $\Phi$ to the state he retrieved, followed by a fixed channel $\Theta$, and an incoherent measurement of his choice. Based on its outcome, he then announces his guess. Assuming that Bob knows $\lambda$ and the phases $\vec{\phi}$, and that he uses the optimal state $\rho$, the best pre-processing $\Phi$, and the optimal incoherent measurement, his probability of guessing correctly is given by~\cite[Prop.~17]{Theurer2019},
	\begin{equation}
		p^{\text{max}}_{\lambda,\vec{\phi}}(\Theta)=\frac{1}{2}+\frac{1}{2}\max_{\Phi\in \mathcal{DI}}\left|\left|\Delta \Theta \Phi \left(\lambda-\mu \Lambda_{\vec{\phi}}\right)\right|\right|_1,
	\end{equation}
	where $\left|\left|\Theta\right|\right|_1$ denotes the induced trace norm of the operation $\Theta$. Here and in the following, we always implicitly assume that the in- and output dimensions of operations and states that we connect fit, i.e., the state $\rho$ that Bob sends to Alice is an element of the Hilbert space on which $\Lambda_{\vec{\phi}}$ acts, and the maximization is understood to run over  all detection incoherent channels $\Phi$ with in- and output spaces determined by the ones of $\Theta$ and $\Lambda_{\vec{\phi}}$.
	
	If Bob had no access to $\Theta$, his measurement outcome would not depend on whether Alice applied $\Lambda_{\vec{\phi}}$ or not, because the combination of $\Phi$ and the incoherent measurement alone is not sensitive to the changes that $\Lambda_{\vec{\phi}}$ induces.
	His best strategy would thus be to bet purely based on his knowledge of $\lambda$. The increase in his probability of guessing correctly that $\Theta$ provides is therefore given by the functionals 
	\begin{equation}
		M_{\lambda,\vec{\phi}}(\Theta):=\max_{\Phi\in \mathcal{DI}}\left|\left|\Delta \Theta \Phi \left(\lambda-\mu \Lambda_{\vec{\phi}}\right)\right|\right|_1-|\lambda-\mu|,\label{eq:pre_proc}
	\end{equation} 
	which we will call \emph{pre-processed improvements}. These functionals define a family of convex measures in the detection incoherent setting, which is the content of the following Theorem. Its proof can be found in App.~\ref{ap:proofs}, where we also provide the other proofs of the results presented in the main text.
	\begin{theorem}\label{theo:meas_pre}
		The functionals $M_{\lambda,\vec{\phi}}(\Theta)$ are convex measures in the detection incoherent setting for all $\lambda\in [0,1]$ and for all $\vec{\phi}\in\mathbbm{R}^M$.
	\end{theorem}
	
	We show in App.~\ref{ap:preproc} that without the optimal pre-processing $\Phi$, the functionals in Eq.~\eqref{eq:pre_proc} are in general not proper measures in the detection incoherent setting. On the other hand, adding a free post-processing after the channel $\Theta$ does not increase Bob's chances of success. This is a direct consequence of the fact that a detection incoherent operation cannot convert an incoherent measurement to a coherent one, and we can thus absorb any free post-processing into the incoherent measurement. Together with the fact that $M_{\lambda,\vec{\phi}}(\Theta)=M_{\lambda,\vec{\phi}}(\Theta\otimes\id)$, which we show in the proof of Thm.~\ref{theo:meas_pre}, this implies that replacing the optimal free pre-processing $\Phi$ by a free pre- and post-processing as well as a memory channel, i.e., 
	\begin{align*}
		\Theta\Phi\rightarrow \Psi(\Theta\otimes\id)\tilde{\Phi}, \quad \tilde{\Phi},\Psi\in \mathcal{DI},
	\end{align*}
	does not increase Bob's chances of success.

	In discrimination games similar to the one that we discuss here, it is frequently the case that the correct usage of auxiliary systems increases the chances of success~\cite[Chap~3.3]{Watrous2018}. 
	Therefore, one might assume that it is beneficial for Bob to prepare a correlated state of a composed system $AZ$ and hand only a part of it, i.e., $A$, to Alice (see Fig.~\ref{fig:diam} for the adapted protocol). 
	\begin{figure}[ht]
		\includegraphics[width=.95\linewidth]{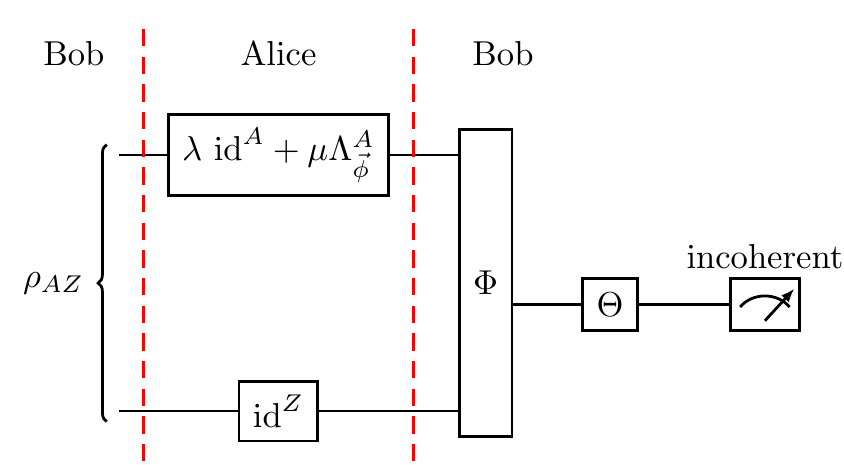}
		\caption{Sketch of a potential method with which Bob might increase his chances in the guessing game represented in Fig.~\ref{fig:op-pre}. He prepares a correlated state and hands only a subsystem to Alice.}\label{fig:diam}
	\end{figure}
	This is however not true.
	\begin{theorem}\label{theo:aux}
		An auxiliary system does not increase the pre-processed improvements, i.e., for $\vec{\tilde{\phi}}(AZ)$ such that 
		\begin{equation*}
			\Lambda_{\vec{\tilde{\phi}}(AZ)}:=\Lambda_{\vec{\phi}}^{A}\otimes\id^{Z},
		\end{equation*}
		it holds that
		\begin{equation*}
			M_{\lambda,\vec{\tilde{\phi}}(AZ)}(\Theta)= M_{\lambda,\vec{\phi}}(\Theta) \quad \forall \Theta\ \forall \lambda \ \forall {\vec{\phi}} \ \forall Z .
		\end{equation*}
	\end{theorem}

	These results imply that the pre-processed improvements describe the maximal usefulness of an operation's ability to detect coherence in our guessing games: if we have access to $\Theta$ and are allowed to combine it with arbitrary operations that cannot detect coherence, but with none that can, the pre-preocessed improvements quantify the advantage that $\Theta$ grants. The pre-processed improvements are thus resource measures with a clear operational interpretation in terms of the games.
	
	Moreover, these guessing games are directly connected to our interferometric setup: Bob's preparation of the arbitrary $\rho$  
	corresponds to a choice of the incoherent $\tau$ as well as $\Theta_1$ in Fig.~\ref{fig:multi-path}, and $\Theta_2$ represents the combination of $\Phi$ and $\Theta$. 
	Since we are only interested in $\Theta_1(\tau)$, but not in $\Theta_1$ alone, we can always choose a pair $\Theta_1$ and $\tau$ such that $\Theta_1$ does not detect coherence  ($\Theta_1(\tau)=\rho\Tr(\tau)$, where $\rho$ is our optimal state).
	The pre-processed improvements $M_{\lambda,\vec{\phi}}(\Theta)$ therefore describe the maximal usefulness of an operation's ability to detect coherence in a concrete interferometric task, i.e., deciding if, e.g., a set of samples was present or not. Since the pre-processed improvements are also valid resource measures in the detection incoherent setting, we showed one of the main results of this paper: the ability to detect coherence is a resource in interferometry.
	
	As we pointed out in Sec.~\ref{subsec:rts}, another desirable property for measures is faithfulness. In our operational setting, it would ensure that every non-free channel is at least a little helpful for the interferometric task that we intend to accomplish. Faithfulness of the functionals $M_{\lambda,\vec{\phi}}(\Theta)$ depends however on $\lambda$, which is the content of the following Theorem.
	\begin{theorem}\label{thm:faithfulnessPre}
		For all $\vec{\phi}\in\mathbbm{R}^M$ with at least two different components, the functionals $M_{\lambda,\vec{\phi}}(\Theta)$ are faithful if and only if $\lambda=\frac{1}{2}$.
	\end{theorem}
	In App.~\ref{ap:peek}, we discuss in more detail why the functionals $M_{\lambda,\vec{\phi}}(\Theta)$ are only faithful for $\lambda=\frac{1}{2}$. In essence, the intuitive reason behind this can be understood on purely classical grounds: if  $\lambda\approx 1$ and the conclusions that can be drawn from the measurement outcomes are very uncertain, the best guessing strategy is determined exclusively by the knowledge of $\lambda$, i.e., by always guessing the a priory more likely case. Moreover, there exist resourceful operations arbitrarily close to the free ones and, using free operations, the measurement outcomes are insensitive to whether $\Lambda_{\vec{\phi}}$ was applied or not. This implies that for $\lambda\ne\frac{1}{2}$, we can always find a resourceful operation such that the optimal guessing strategy only depends on $\lambda$ (which we exploit in the proof of this Theorem given in App.~\ref{ap:proofs}). However, since every resourceful operation can detect coherence and the measurement outcomes thus depend on whether $\Lambda_{\vec{\phi}}$ was applied or not (given that Bob plays the game ideally), $M_{\frac{1}{2},\vec{\phi}}(\Theta)$ is faithful.

	Now that we connected the pre-processed improvements with interferometry, a natural question to ask is how to evaluate them numerically. This is for example relevant if we want to decide if one operation outperforms another one in our discrimination games. Since the pre-processed improvements are defined via the optimizations in Eq.~\eqref{eq:pre_proc} and the induced trace norm includes an additional optimization over states, evaluating $M_{\lambda,\vec{\phi}}(\Theta)$ is however not straightforward. In App.~\ref{ap:eval}, we propose a method based on semidefinite programming that in addition leads to an optimal state and pre-processing.

	\subsection{Creating coherence and interferometry}
	Here, we analyse a setting that connects the creation of coherence to interferometry. As in Sec.~\ref{subsec:DetMain}, we introduce it as a game between Alice and Bob (see Fig.~\ref{fig:op-post}). 
	This time, Bob is provided with a fixed \emph{creating operation} $\Theta$ that he applies to an incoherent state of his choice. Before sending this state to Alice, Bob is allowed to further apply an arbitrary creation incoherent operation $\Psi$ onto $\Theta(\tau)$. Alice receives the state $\Psi\Theta(\tau)$, applies again with probability $\mu$ a fixed operation $\Lambda_{\vec{\phi}}$, and sends the resulting state back to Bob. Finally, Bob performs a generic measurement on the state he retrieved and guesses whether Alice applied $\Lambda_{\vec{\phi}}$ or not.

	\begin{figure}[ht]
		\includegraphics[width=.95\linewidth]{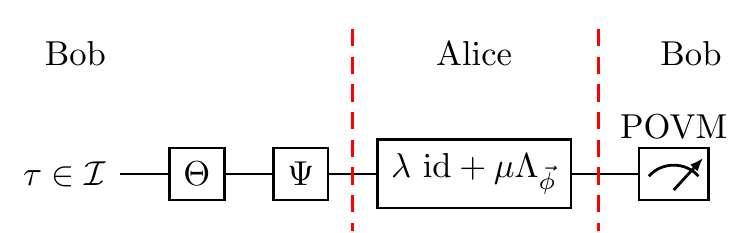}
		\caption{Sketch of the game played by Alice and Bob that is used to describe the role of the creation of coherence in interferometry.}\label{fig:op-post}
	\end{figure}
	
	In this case, assuming that Bob prepares the best incoherent state $\tau$, applies the best post-processing $\Psi$, and the best final measurement, his probability of guessing correctly whether Alice applied $\Lambda_{\vec{\phi}}$ or not is given by~\cite{Watrous2018}
	\begin{equation}
		p^{\text{max}}_{\lambda,\vec{\phi}}(\Theta)=\frac{1}{2}+ \frac{1}{2}\max_{\Psi\in\mathcal{MIO}}\left|\left|\left(\lambda-\mu \Lambda_{\vec{\phi}}\right) \Psi \Theta\Delta\right|\right|_1.
	\end{equation} 
	Analogously to the previous setting, we define the functionals 
	\begin{equation}
		N_{\lambda,\vec{\phi}}(\Theta):=\max_{\Psi\in\mathcal{MIO}}\left|\left|\left(\lambda-\mu \Lambda_{\vec{\phi}}\right) \Psi \Theta\Delta\right|\right|_1-|\lambda-\mu |.\label{eq:post_proc}
	\end{equation}
	and call them \emph{post-processed improvements}. The connection to our interferometric setup is again straight forward: the operation $\Theta_1$ in Fig.~\ref{fig:multi-path} is represented by the joint action of $\Theta$ and $\Psi$, and $\Theta_2$ and the incoherent measurement afterwards form the general POVM. Analogously to Thm.~\ref{theo:meas_pre}, we further find
	
	\begin{theorem}\label{theo:meas_post}
		The functionals $N_{\lambda,\vec{\phi}}(\Theta)$ are convex measures in the creation incoherent setting for all $\lambda\in [0,1]$ and for all $\vec{\phi}\in\mathbbm{R}^M$.
	\end{theorem}

	The post-processed improvements therefore quantify the ability of an operation $\Theta$ to create coherence and, using similar arguments as in the previous section, have the operational interpretation that they describe the advantage that $\Theta$'s ability to create coherence grants in a concrete interferometric setup. We thus established another main result, namely that the ability to create coherence is a relevant resource in interferometry too.
	
	Concerning faithfulness, we have the following result.
	\begin{theorem}\label{theo:faith_post}
		The functionals $N_{\frac{1}{2},\vec{\phi}}(\Theta)$ are faithful for all $\vec{\phi}\in\mathbbm{R}^M$ with at least two different components.
	\end{theorem}
	This ensures again that the ability to create coherence contributes in some interferometric tasks, namely in the ones described by our games with $\lambda=\mu=\frac{1}{2}$.

	\section{Conclusions}\label{sec:conc}
	In this work, we introduced families of dynamical resource measures that allowed us to establish a connection between an operation's ability to detect or create coherence and the performance of interferometric experiments. This shows that the abstract resource theories defined in Ref.~\cite{Theurer2019} have  an operational meaning. Our results concerning the ability of operations to create coherence should be compared to the static results of Ref.~\cite{Biswas2017}, where it was shown  that every visibility functional that satisfies some meaningful properties can be used to define a coherence measure that is strongly monotonic~\cite{Baumgratz2014} under strictly incoherent operations~\cite{Winter2016,Lami2019,Lami2020b}. One obtains these static coherence measures from the visibilities via optimizations over all measurements. Similar to our case, such optimizations are necessary to ensure that the coherence is used ideally and not only present. 
	Here, we took a more direct approach that did not rely on visibilities. This allowed us to define measures that are not restricted to strictly incoherent operations, but hold for the larger class of maximally incoherent operations instead. An interesting question is whether one could also define dynamical resource measures based on visibility, and whether this would lead back to strictly incoherent operations. In this context, see also Ref.~\cite{Paul2017}, which connected coherence as measured by a normalized version~\cite{Bera2015} of the $l_1$-norm of coherence~\cite{Baumgratz2014} with a visibility in multislit interference. 
	
	One should also compare our results to the role that measures of robustness~\cite{Vidal1999} (weight~\cite{Lewenstein1998}) play in discrimination (exclusion) games~\cite{Piani2016,Napoli2016,Takagi2019a,Skrzypczyk2019a,Skrzypczyk2019b,Uola2019,Mori2020,Ducuara2020a,Uola2020,Ducuara2020b}: whilst our families of measures quantify the advantage that a resource grants in a specific binary discrimination or equivalently exclusion game, measures of robustness (weight) describe the achievable advantage that resourceful objects (e.g., states, measurements, channels, or combinations thereof) give in an ideal discrimination (exclusion) game over free objects. Technically, this is another way to ensure that the resources are used appropriately.
	
	Whilst we provided a method to compute the pre-processed improvements and showed that they are not faithful for $\lambda\ne\mu$, it is an open question whether there exist analogous results for the post-processed cases. Moreover, one could combine the two resource theories we applied and consider a fixed operation used for the creation of coherence and one for its detection. Potentially, the resulting success probabilities could then be expressed as products of two measures. 
	Another interesting idea is to remove, e.g., the optimal pre-processing in our measure, and require that Alice applies an optimal $\Lambda_{\vec{\phi}}$ instead. Whilst we investigated this approach too, we were not able to prove monotonicity (for input dimensions of $\Theta$ greater than three). 
	
	In our investigations, we considered the application of fixed phases that are known to Bob. Whilst this is certainly a relevant scenario, e.g., if one checks whether a known sample is present or not, one often uses interferometers to gather information about unknown phases. It is then an open question whether one could use, e.g., Fisher information, to construct measures in these scenarios~\cite{Biswas2017,Feng2017,Tan2018}. In conclusion, the investigation of the technological relevance of dynamical coherence is far from being completed, but our proof of principle shows that these theories might help to better understand and thus exploit quantum properties.

	\begin{acknowledgments}
		We thank Mirko Rossini, Dario Egloff, and Ludovico Lami for discussions and feedback. 
	\end{acknowledgments}
	
	\appendix
	
	\section{Technical results}\label{ap:index}
	In this Appendix, we collect some technical results that are needed for the proofs of the results in the main text presented in App.~\ref{ap:proofs}.
	In both Appendices, we often represent the action of a quantum channel $\,\Theta$ on matrix elements $\ketbra{i}{j}$ as 
	\begin{equation}\label{eq:ind_ch}
		\Theta(\ketbra{i}{j})=\sum_{k,l}\Theta^{i,j}_{k,l}\ketbra{k}{l}.
	\end{equation} 
	With this notation, we have~\cite[Prop.~15]{Theurer2019}
	\begin{equation}\label{eq:indexCreIn}
		\Theta^{i,i}_{k,l}\propto\delta_{k,l}\,\forall i,k,l,
	\end{equation}
	for all creation incoherent operations $\Theta$, whilst for an operation that cannot detect coherence
	\begin{equation}\label{eq:indexDetIn}
		\Theta^{i,j}_{k,k}\propto\delta_{i,j}\,\forall i,j,k
	\end{equation}
	is satisfied.
	Moreover, the following Proposition holds.
	
	\begin{prop}\label{prop:ind_pr} If $\Theta$ is a quantum channel, the corresponding coefficients $\Theta^{i,j}_{k,l}$ fulfill the following properties:
		\begin{itemize}
			\item[1.] $\Theta^{n,n}_{m,m}\geq 0\quad \forall m,n$,
			\item[2.] $\Theta^{i,j}_{k,l}=\Theta^{j,i*}_{l,k}\quad\forall i,j,k,l$,
			\item[3.] $\sum_{m}\Theta_{m,m}^{i,j}=\delta_{i,j}\quad \forall i, j,m$.
		\end{itemize}
	\end{prop}
	
	\begin{proof}
		The first claim follows from complete positivity. Let us write the Choi state corresponding to the quantum operation $\Theta$ as
		\begin{equation*}
			J_\Theta=\Theta\otimes\id\sum_{i,j}\ketbra{ii}{jj}=\sum_{i,j,k,l}\Theta^{i,j}_{k,l}\ketbra{ki}{lj}.
		\end{equation*}
		Complete positivity of $\Theta$ is then equivalent to $\bra{v}J_\Theta\ket{v}\geq 0 \quad \forall\ket{v}$. Choosing $\ket{v}=\ket{nm}$, we hence find 
		\begin{equation*}
			\bra{mn}J_\Theta\ket{mn}=\Theta^{n,n}_{m,m}\geq 0 \quad \forall m,n
		\end{equation*}
		as a necessary but not sufficient condition for complete positivity.
		
		The second claim follows from the fact that a quantum operation preserves hermiticity. The hermitian conjugate of our Choi matrix is given by
		\begin{equation*}
			J_\Theta^\dagger=\sum_{i,j,k,l}\Theta^{i,j*}_{k,l}\ketbra{lj}{ki}=\sum_{i,j,k,l}\Theta^{j,i*}_{l,k}\ketbra{ki}{lj}.
		\end{equation*}
		Now, since $J_\Theta^\dagger=J_\Theta$ must hold, we have $\Theta^{j,i*}_{l,k}=\Theta^{i,j}_{k,l}\quad\forall i,j,k,l$.
		
		The third claim follows from trace preservation. Let us write a state as $\rho=\sum_{i,j}\rho_{i,j}\ketbra{i}{j}$ and the action of the channel $\Theta$ on it as
		\begin{equation*}
			\Theta(\rho)=\sum_{i,j,k,l}\rho_{i,j}\Theta^{i,j}_{k,l}\ketbra{k}{l}.
		\end{equation*}
		Choosing $\rho$ such that $\rho_{n,n}=1$ for a fixed $n$ and $\rho_{i,j}=0$ for all other coefficients, we find the following necessary condition for trace preservation
		\begin{equation*}
			1=\Tr\left(\Theta(\rho)\right)=\sum_{m,i,j}\rho_{i,j}\Theta^{i,j}_{m,m}=\sum_m\Theta^{n,n}_{m,m} \quad \forall n.
		\end{equation*}
		
		The case for $i\neq j$ can be proven by changing the choice of $\rho$. Let us take a state such that $\rho_{\tilde{i},\tilde{i}}=\rho_{\tilde{j},\tilde{j}}=\frac{1}{2}$ and $\rho_{\tilde{i},\tilde{j}}=\frac{1}{2}e^{i\xi}=\rho_{\tilde{j},\tilde{i}}^*$ for a pair of indices $\tilde{i}\ne\tilde{j}$. This automatically implies that all other coefficients are zero. Thus, using again trace preservation and the previous result, we find
		\small
		\begin{equation*}
			\sum_m\frac{1}{2}\!\left(\Theta_{m,m}^{\tilde{i},\tilde{i}}+\Theta_{m,m}^{\tilde{j},\tilde{j}}\right)+\sum_m\frac{1}{2}\!\left(\Theta_{m,m}^{\tilde{i},\tilde{j}}e^{i\xi}+\Theta_{m,m}^{\tilde{j},\tilde{i}}e^{-i\xi}\right)=1\  \forall\xi
		\end{equation*}
		\normalsize
		\begin{equation*}
			\iff \Re\left(e^{i\xi}\sum_m\Theta_{m,m}^{\tilde{i},\tilde{j}}\right)=0 \quad \forall\xi
		\end{equation*}
		\begin{equation*}
			\iff \sum_m\Theta_{m,m}^{\tilde{i},\tilde{j}}=0.
		\end{equation*}
	\end{proof}
	
	The above results now allow us to prove the following Lemma, where we use the notation introduced below Eq.~\eqref{eq:super}: when we need to specify that a quantum channel $\Theta$ has an input system $A$ and an output system $B$, we will write  $\Theta^{B\leftarrow A}$ and simply $\Theta^A$ when input and output system are identical.
	\begin{lem}\label{lem:tensor}
		An operator 
		\begin{equation}
			\alpha=\left(\id^A\otimes\Delta^B\right) \Theta^{AB\leftarrow C}_{\mathcal{DI}} \left(\lambda-\mu\Lambda_{\vec{\phi}}\right)(\rho),
		\end{equation} where $\Theta^{AB\leftarrow C}_{\mathcal{DI}}$ is a detection incoherent quantum channel, can be decomposed as 
		\begin{equation}
			\alpha=\sum_bp_b\left(\tilde{\Theta}^{A\leftarrow C}_{\mathcal{DI},b} \left(\lambda-\mu\Lambda_{\vec{\phi}}\right)(\rho)\right)\otimes\ketbra{b}{b}_B,
		\end{equation} where $\tilde{\Theta}^{A\leftarrow C}_{\mathcal{DI},b}$ are detection incoherent quantum channels and $p_b$ are probabilities.
	\end{lem}
	\begin{proof}
		Let us write down $\alpha$ using a Kraus representation of $\left(\id^A\otimes\Delta^B\right) \Theta^{AB\leftarrow C}_{\mathcal{DI}}$, i.e.,
		\begin{equation*}
			\alpha=\sum_{n,b}(\mathbbm{1}_A\otimes\ketbra{b}{b}_B) K_n \left(\lambda-\mu \Lambda_{\vec{\phi}}\right)(\rho) K_n^\dagger (\mathbbm{1}_A\otimes\ketbra{b}{b}_B).
		\end{equation*}
		Now, we define the operator
		\begin{equation*}
			\alpha_{|b}:= \bra{b}_B \alpha \ket{b}_B=\sum_n\bra{b}_B K_n\left(\lambda-\mu\Lambda_{\vec{\phi}}\right)(\rho) K_n^\dagger\ket{b}_B
		\end{equation*}
		and evaluate its trace using the index representation, Eq.~\eqref{eq:indexDetIn}, and Prop.~\ref{prop:ind_pr},
		\begin{align*}
			\Tr \left(\alpha_{|b}\right)&=\sum_a \bra{a,b}_{AB}\Theta^{AB\leftarrow C}_{\mathcal{DI}} \left(\lambda-\mu\Lambda_{\vec{\phi}}\right)(\rho) \ket{a,b}_{AB} \\
			&=\sum_a \bra{a,b}_{AB} \Bigg(\sum_{i,j,k,l,m,n}\Theta^{i,j}_{kl,mn} \\
			& \pushright{\left(\lambda-\mu e^{i(\phi_i-\phi_j)}\right)\rho_{i,j} \ketbra{kl}{mn}_{AB}\Bigg) \ket{a,b}_{AB}}\\
			&=(\lambda-\mu )\sum_a \sum_{i}\Theta^{i,i}_{ab,ab}\rho_{i,i} \\
			&=(\lambda-\mu )\Tr\left(\mathbbm{P}_b\Theta^{AB\leftarrow C}_{\mathcal{DI}}(\rho)\right)\\
			&=(\lambda-\mu)p_b,
		\end{align*} 
		where $p_b$ is the probability of collapsing the state $\Theta^{AB\leftarrow C}_{\mathcal{DI}}(\rho)$ to the subspace onto which the operator 
		\begin{equation*}
			\mathbbm{P}_b=\sum_a\ketbra{a,b}{a,b}_{AB}
		\end{equation*} projects. 
		At this point, we define $\tilde{K}_{n,b}:=\frac{\bra{b}_B K_n}{\sqrt{p_b}}\ \forall b$ such that $p_b\neq 0$. This allows us to write
		\begin{align*}
			\alpha &=\sum_b \alpha_{|b}\otimes\ketbra{b}{b}_B \\
			&=\sum_{b\ :\ p_b\neq 0} p_b\sum_n\tilde{K}_{n,b}\left(\lambda-\mu\Lambda_{\vec{\phi}}\right)(\rho)\tilde{K}_{n,b}^\dagger\otimes\ketbra{b}{b}_B.
		\end{align*}
		For all $b$ with $p_b\ne0$, we now interpret $\{\tilde{K}_{n,b}\}_{n}$ a as set of Kraus operators associated to an operation $\tilde{\Theta}^{A\leftarrow C}_{\mathcal{DI},b}$, which is thus hermiticity preserving and completely positive. Trace preservation of $\tilde{\Theta}^{A\leftarrow C}_{\mathcal{DI},b}$ follows from
		\begin{align*}
			&\Tr  \left(\tilde{\Theta}^{A\leftarrow C}_{\mathcal{DI},b}(\rho)\right)\nonumber\\
			&=\frac{1}{p_b}\Tr\Bigg(\left(\mathbbm{1}_A\otimes\bra{b}_B\right) \sum_n K_n\rho K_n^\dagger\left(\mathbbm{1}_A\otimes\ket{b}_B\right)\Bigg)\\
			&=\frac{1}{p_b}\sum_a\bra{a}_A\otimes\bra{b}_B\sum_n K_n\rho K_n^\dagger\ket{a}_A\otimes\ket{b}_B\\
			&=\frac{1}{p_b}\Tr\left(\mathbbm{P}_b\Theta^{AB\leftarrow C}_{\mathcal{DI}}(\rho)\right)=\frac{p_b}{p_b}=1.
		\end{align*} 
		It remains to prove that that $\Theta^{A\leftarrow C}_{\mathcal{DI},b}$ is detection incoherent. From
		\begin{align*}
			\tilde{\Theta}^{A\leftarrow C}_{\mathcal{DI},b}(\rho)&=\frac{\bra{b}_B}{\sqrt{p_b}}\Theta^{AB\leftarrow C}_{\mathcal{DI}}(\rho)\frac{\ket{b}_B}{\sqrt{p_b}} \\
			&=\frac{1}{p_b}\bra{b}_B\left(\sum_{i,j,k,l,m,n}\Theta^{i,j}_{kl,mn}\rho_{i,j}\ketbra{kl}{mn}\right)\ket{b}_B\\
			&=\sum_{i,j,k,m}\frac{\Theta^{i,j}_{kb,mb}\rho_{i,j}}{p_b}\ketbra{k}{m}_A \\
			&=\sum_{i,j,k,l}\tilde{\Theta}^{i,j}_{k,l}(b)\rho_{i,j}\ketbra{k}{l}_A
		\end{align*} 
		follows that $\tilde{\Theta}^{i,j}_{k,k}(b)=\frac{\Theta^{i,j}_{kb,kb}}{p_b}\propto\delta_{i,j}$. Due to Ref.~\cite[Prop.~15]{Theurer2019} the $\tilde{\Theta}^{A\leftarrow C}_{\mathcal{DI},b}$ are thus detection incoherent because $\Theta^{AB\leftarrow C}_{\mathcal{DI}}$ was.
	\end{proof}
	In addition, we will need the following lemmas to prove the results in the main text.
	\begin{lem}\label{lem:detect_two}
		Consider an operation $\Theta^{B\leftarrow A}$ with $\dim(B)=2$. If there exists an $m$ such that $\Theta^{i,j}_{m,m}=0 \ \forall i,j$, then $\Theta^{i,j}_{k,l}=0 \ \forall k\neq l, \forall i,j $.
	\end{lem}
	
	\begin{proof}
		Let us assume that $m=0$. If this is not the case we just need to relabel our Hilbert space. Applying $\Theta^{B\leftarrow A}$ to a state $\rho$, we get
		\begin{align*}
			\Theta(\rho)&=\sum_{i,j}\Theta^{i,j}_{0,0}\rho_{i,j}\ketbra{0}{0}+\sum_{i,j}\Theta^{i,j}_{1,1}\rho_{i,j}\ketbra{1}{1} \\
			&\pushright{ +\sum_{i,j}\sum_{m\neq n}\Theta^{i,j}_{m,n}\rho_{i,j}\ketbra{m}{n}} \\
			&=\sum_{i,j}\Theta^{i,j}_{1,1}\rho_{i,j}\ketbra{1}{1}+\sum_{i,j}\Theta^{i,j}_{0,1}\rho_{i,j}\ketbra{0}{1}\\
			&\pushright{+\sum_{i,j}\Theta^{i,j}_{1,0}\rho_{i,j}\ketbra{1}{0}.}
		\end{align*} 
		At this point, we recall that $\Theta(\rho)$ must be positive semidefinite and that a necessary condition for positive semidefiniteness of an operator is that its determinant is non-negative. Moreover, this must hold for all states $\rho$. Using Proposition~\ref{prop:ind_pr}.2, we get
		\begin{equation*}
			\det(\Theta(\rho))=-\left|\sum_{i,j}\Theta^{i,j}_{1,0}\rho_{i,j}\right|^2.
		\end{equation*} 
		We now consider two particular choices of states for which the previous condition must hold. The first one is $\rho=\ketbra{i}{i}$. With this, we obtain
		\begin{equation*}
			\det(\Theta(\rho))=-\left|\Theta^{i,i}_{1,0}\right|^2 \geq 0 \ \forall i \quad\iff\quad\Theta^{i,i}_{1,0}=0 \ \forall i.
		\end{equation*}
		The second choice is 
		\begin{equation*}
			\rho=\frac{1}{2}\left(\ketbra{i}{i}+\ketbra{j}{j}+\ketbra{i}{j}e^{i\xi}+\ketbra{j}{i}e^{-i\xi}\right).
		\end{equation*} We get
		\begin{align*}
			\det(\Theta(\rho))&=-\frac{1}{2}\left|\Theta^{i,i}_{1,0}+\Theta^{j,j}_{1,0}+\Theta^{i,j}_{1,0}e^{i\xi}+\Theta^{j,i}_{1,0}e^{-i\xi}\right|^2 \\
			&=-\frac{1}{2}\left|\Theta^{i,j}_{1,0}e^{i\xi}+\Theta^{j,i}_{1,0}e^{-i\xi}\right|^2 \geq 0 \ \forall i\neq j,\ \forall \xi \\
			&\iff\quad\Theta^{i,j}_{1,0}=0 \ \forall i\neq j.
		\end{align*}
		Hence, $\Theta^{i,j}_{1,0}=0 \ \forall i,j$. Using again Proposition~\ref{prop:ind_pr}.2 finishes the proof. 
	\end{proof}
	
	\begin{lem}\label{lem:free_exp}
		For all states $\ket{\psi}_{AZ}$ there exist a detection incoherent operation $\Phi^{AZ\leftarrow A}$ and a state $\ket{\varphi}_A$ such that 
		\begin{align}
			\Phi^{AZ\leftarrow A} &\left(\lambda-\mu \Lambda_{\vec{\phi}}^A\right)([\ketbra{\varphi}{\varphi}_{A}) \nonumber \\
			&=\left(\left(\lambda-\mu \Lambda_{\vec{\phi}}^A\right)\otimes\id^{Z}\right)(\ketbra{\psi}{\psi}_{AZ}) . \label{eq:isometry}
		\end{align} 
	\end{lem}
	
	\begin{proof}
		Given $\ket{\psi}_{AZ}:=\sum_{m,n}\psi_{m,n}\ket{m,n}_{AZ}$, we choose $\ket{\varphi}_A$ such that
		\begin{equation*}
			\ket{\varphi}_A:=\sum_m\varphi_m\ket{m}_A \quad\text{with}\quad|\varphi_m|^2=\sum_n|\psi_{m,n}|^2 \ \forall m.
		\end{equation*}
		Next, we define the operator 
		\begin{equation*}
			U:= \sum_{m,n}v_{m,n}\ket{m,n}_{AZ}\bra{m}_A,
		\end{equation*} where
		\begin{align*}
			&v_{m,n}=\frac{\psi_{m,n}}{\varphi_m} &\forall m \quad \text{s.t.} \quad \varphi_m\neq 0,\quad\forall n, \\
			&v_{m,n}= \frac{1}{\dim(Z)} &\forall m \quad \text{s.t.} \quad \varphi_m= 0,\quad \forall n,
		\end{align*} 
		and find that 
		\begin{align*}
			\sum_n|v_{m,n}|^2=\sum_n\frac{|\psi_{m,n}|^2}{|\varphi_m|^2}=1 \quad &\forall m \quad \text{s.t.} \quad \varphi_m\neq 0, \\
			\sum_n|v_{m,n}|^2=\sum_n\frac{1}{\dim(Z)}=1 \quad &\forall m \quad \text{s.t.} \quad \varphi_m= 0.
		\end{align*} 
		This ensures that 
		\begin{equation*}
			U^\dagger U=\sum_m\left(\sum_n|v_{m,n}|^2\right)\ketbra{m}{m}_A=\mathbbm{1}_A,
		\end{equation*} hence $U$ is an isometry and it defines a CPTP map. 
		
		Let us now verify that this map is detection incoherent. Given a state $\omega_A=\sum_{i,j}\omega_{i,j}\ketbra{i}{j}_A$, we have
		\begin{align*}
			\Delta &\left(U\Delta(\omega_A)U^\dagger\right)=\Delta\Bigg(\left(\sum_{m,n}v_{m,n}\ket{m,n}_{AZ}\bra{m}_A\right) \\
			&\pushright{\left(\sum_{i}\omega_{i,i}\ketbra{i}{i}_A\right)\left(\sum_{o,p}v_{o,p}\ket{o,p}_{AZ}\bra{o}_A\right)^\dagger\Bigg)} \\
			&=\Delta\left(\sum_{i,n,p}v_{i,n}\omega_{i,i}v_{i,p}^*\ketbra{in}{ip}_{AZ} \right) \\
			&=\sum_{i,n}|v_{i,n}|^2\omega_{i,i}\ketbra{in}{in}_{AZ},
		\end{align*}while
		\begin{align*}
			\Delta &\left(U\omega_A U^\dagger\right)=\Delta\Bigg(\left(\sum_{m,n}v_{m,n}\ket{m,n}_{AZ}\bra{m}_A\right) \\
			&\qquad\left(\sum_{i,j}\omega_{i,j}\ketbra{i}{j}_A\right)\left(\sum_{o,p}v_{o,p}\ket{o,p}_{AZ}\bra{o}_A\right)^\dagger\Bigg) \\
			&=\Delta\left(\sum_{i,j,n,p}v_{i,n}\omega_{i,j}v_{j,p}^*\ketbra{in}{jp}_{AZ} \right) \\
			&=\sum_{i,n}|v_{i,n}|^2\omega_{i,i}\ketbra{in}{in}_{AZ}=\Delta\left(U\Delta(\omega_A)U^\dagger\right).
		\end{align*} Therefore, the map defined by $U$ is detection incoherent. Let us call this map $\Phi^{AZ\leftarrow A}$ and use it to show Eq.~(\ref{eq:isometry}). We obtain
		\begin{align*}
			\Phi^{AZ\leftarrow A} &\left(\lambda-\mu \Lambda_{\vec{\phi}}^A\right)((\ketbra{\varphi}{\varphi}_{A}) \\
			&=U\left(\sum_{i,j}\left(\lambda-\mu e^{i(\phi_i-\phi_j)}\right)\varphi_i\varphi_j^*\ketbra{i}{j}_A\right)U^\dagger \\
			&=\sum_{i,j,n,p}\left(\lambda-\mu e^{i(\phi_i-\phi_j)}\right)\varphi_i\varphi_j^*v_{i,n}v_{j,p}^*\ketbra{in}{jp}_{AZ} \\
			&=\sum_{i,j,n,p}\left(\lambda-\mu e^{i(\phi_i-\phi_j)}\right)\psi_{i,n}\psi_{j,p}^*\ketbra{in}{jp}_{AZ} \\
			&=\left(\left(\lambda-\mu \Lambda_{\vec{\phi}}^A\right)\otimes\id^{Z}\right)(\ketbra{\psi}{\psi}_{AZ}),
		\end{align*}
		where the third equality follows from the definition of $v_{m,n}$. We recall that $\varphi_m=0$ for some $m$ iff for that $m$ we had $\psi_{m,n}=0 \quad \forall n$.
	\end{proof}

	\begin{lem}\label{lem:post_tensor}
		Let $\Psi^{C\leftarrow BZ},\Theta^{B\leftarrow A}$ be  quantum channels and  $\rho_{AZ}$ an incoherent state. Then there exists a probability distribution $\{p_i\}_i$, incoherent states $\{\rho_{|i}\}_i$, and quantum channels $\{\Psi_i^{C\leftarrow B}\}_i$ such that 
		\begin{equation*}
			\Psi^{C\leftarrow BZ}\!\left(\Theta^{B\leftarrow A}\otimes\id^Z\right)(\rho_{AZ})=\sum_i p_i\Psi_i^{C\leftarrow B} \Theta^{B\leftarrow A}(\rho_{|i}).
		\end{equation*} 
		Moreover, if $\Psi^{C\leftarrow BZ}\in\mathcal{MIO}$, one can choose $\Psi_i^{C\leftarrow B}\in\mathcal{MIO}$ too.
	\end{lem}
	
	\begin{proof}		
		The probability to obtain outcome $b$ if subsystem $Z$ of the state $\rho_{AZ}$ is projectively measured in its incoherent basis is given by
		\begin{align*}
			p_b=\Tr\left\{\bra{b}_Z\rho_{AZ}\ket{b}_Z\right\}, 
		\end{align*} 
		and, for $p_b\neq 0$, we denote the corresponding post measurement states of system $A$ by
		\begin{equation*}
			\rho_{|b}:=\frac{\bra{b}_Z\rho_{AZ}\ket{b}_Z}{p_b}.
		\end{equation*} 
		Since $\rho_{AZ}\in\mathcal{I}$ by assumption, the $\rho_{|b}$ are incoherent too and we have
		\begin{align*}
			&\left(\Theta^{B\leftarrow A}\otimes\id^Z\right) (\rho_{AZ}) \\
			&=\left(\id^B\otimes\Delta^Z\right)\left(\Theta^{B\leftarrow A}\otimes\id^Z\right)(\rho_{AZ}) \\
			&=\sum_b(\mathbbm{1}_B\otimes\ketbra{b}{b}_Z)\left(\Theta^{B\leftarrow A}\otimes\id^Z(\rho_{AZ})\right)\left(\mathbbm{1}_B\otimes\ketbra{b}{b}_Z\right) \\
			&=\sum_b\bra{b}_Z\left(\Theta^{B\leftarrow A}\otimes\id^Z(\rho_{AZ})\right)\ket{b}_Z\otimes\ketbra{b}{b}_Z \\
			&=\sum_b\Theta^{B\leftarrow A}\left(\bra{b}_Z\rho_{AZ}\ket{b}_Z\right)\otimes\ketbra{b}{b}_Z \\
			&=\sum_{b\ :\ p_b\neq 0}p_b\Theta^{B\leftarrow A}\left(\rho_{|b}\right)\otimes\ketbra{b}{b}_Z.
		\end{align*} 
		
		Now let $\{K_n\}_n$ be a set of Kraus operators corresponding to $\Psi^{C\leftarrow BZ}$ and define $L_{n,b}:= K_n\ket{b}_Z$. Since $\Psi^{C\leftarrow BZ}$ is a channel, we have $\sum_n K_n^\dagger K_n=\mathbbm{1}_{BZ}$, and hence
		\begin{align*}
			\sum_n L_{n,b}^\dagger L_{n,b}&=\sum_n\bra{b}_Z K_n^\dagger K_n\ket{b}_Z \\
			&=\bra{b}_Z\mathbbm{1}_{BZ}\ket{b}_Z=\mathbbm{1}_B.
		\end{align*} 
		For fixed $b$, the set $\{ L_{n,b}\}_n$ is thus a Kraus decomposition of a channel $\Psi_b^{C\leftarrow B}$, and we can write
		\begin{align*}
			&\Psi^{C\leftarrow BZ} \left(\Theta^{B\leftarrow A}\otimes\id^Z\right)(\rho_{AZ}) \\
			&=\sum_{b\ :\ p_b\neq 0}p_b\Psi^{C\leftarrow BZ} \left(\Theta^{B\leftarrow A}\left(\rho_{|b}\right)\otimes\ketbra{b}{b}_Z\right) \\
			&=\sum_{b\ :\ p_b\neq 0}p_b\sum_{n} K_n\left(\Theta^{B\leftarrow A}\left(\rho_{|b}\right)\otimes\ketbra{b}{b}_Z\right) K_n^\dagger \\
			&=\sum_{b\ :\ p_b\neq 0}p_b\sum_{n} L_{n,b}\left(\Theta^{B\leftarrow A}\left(\rho_{|b}\right)\right) L_{n,b}^\dagger \\
			&=\sum_{b\ :\ p_b\neq 0}p_b\Psi_b^{C\leftarrow B} \Theta^{B\leftarrow A}(\rho_{|b}),
		\end{align*} which completes the first part of the  proof. 
		
		If $\Psi^{C\leftarrow BZ}\in\mathcal{MIO}$, we further obtain
		\begin{align*}
			&\Psi_b^{C\leftarrow B} \Delta(\rho_B)=\sum_n K_n\ket{b}_Z\Delta\left(\rho_B\right)\bra{b}_Z K_n^\dagger \\
			&=\sum_n K_n\Delta\left(\rho_B\otimes\ketbra{b}{b}_Z\right) K_n^\dagger \\
			&=\Psi^{C\leftarrow BZ}\Delta\left(\rho_B\otimes\ketbra{b}{b}_Z\right) \\
			&=\Delta \Psi^{C\leftarrow BZ}\Delta\left(\rho_B\otimes\ketbra{b}{b}_Z\right) \\
			&=\Delta \Psi_b^{C\leftarrow B} \Delta(\rho_B).\\
		\end{align*}
	\end{proof}
	
	\section{Proofs of the results in the main text}\label{ap:proofs}
	\setcounter{theorem}{0}
	Here we collect the proofs of the results in the main text, which we repeat for readability. 
	
	\begin{theorem}
		The functionals $M_{\lambda,\vec{\phi}}(\Theta)$ are convex measures in the detection incoherent setting for all $\lambda\in [0,1]$ and for all $\vec{\phi}\in\mathbbm{R}^M$.
	\end{theorem}
	
	\begin{proof}
		
		Let us start by proving nullity, i.e., we assume that $\Theta$ is detection incoherent. Then, since the composition of free operations is a free operation, we have $\Delta \Theta \Phi=\Delta \Theta \Phi \Delta$.
		
		Defining the \emph{complementary dephasing operator} as $\Delta_c:=\id-\Delta$, we can write
		\begin{align*}
			\Delta \left(\lambda-\mu\Lambda_{\vec{\phi}}\right)(\rho) & =\Delta \left(\lambda-\mu\Lambda_{\vec{\phi}}\right) (\Delta+\Delta_c)(\rho) \\
			& =\Delta (\lambda-\mu) \Delta(\rho)+\Delta \left(\lambda-\mu\Lambda_{\vec{\phi}}\right) \Delta_c(\rho)\\
			&= (\lambda-\mu)\Delta(\rho).
		\end{align*}
		Here, the second line follows from the fact that $\Lambda_{\vec{\phi}}$ does not act on diagonal elements, while the third is due to the joint action of $\Delta_c$ and $\Delta$ which cancel firstly the diagonal and then the off-diagonal terms.
		Therefore, we find
		\begin{align*}
			M_{\lambda,\vec{\phi}}(\Theta)&= \max_{\rho,\Phi}\Tr\left|\Delta \Theta \Phi  \left(\lambda-\mu \Lambda_{\vec{\phi}}\right)(\rho)\right|-\left|\lambda-\mu\right| \\
			&= \max_{\rho,\Phi}\Tr\left|\Delta \Theta \Phi  \Delta \left(\lambda-\mu \Lambda_{\vec{\phi}}\right)(\rho)\right|-\left|\lambda-\mu \right| \\
			&= \left|\lambda-\mu\right| \max_{\rho,\Phi}\Tr\left|\Delta \Theta \Phi \Delta(\rho)\right|-\left|\lambda-\mu \right| \\
			&=\left|\lambda-\mu\right| -\left|\lambda-\mu \right|=0,
		\end{align*}
		where the last line is due to complete positivity and trace preservation of the operations $\Delta$, $\Theta$, and $\Phi$.
		
		Non-negativity follows from 
		\begin{align*}
			M_{\lambda,\vec{\phi}}(\Theta)&=\max_{\rho,\Phi}\Tr\left|\Delta \Theta  \Phi \left(\lambda-\mu \Lambda_{\vec{\phi}}\right)(\rho)\right|-\left|\lambda-\mu\right| \\
			&\geq \max_{\rho,\Phi}\left|\Tr\left(\Delta \Theta  \Phi \left(\lambda-\mu \Lambda_{\vec{\phi}}\right)(\rho)\right)\right|-\left|\lambda-\mu\right| \\
			&=\left|\lambda-\mu\right|-\left|\lambda-\mu\right|=0.
		\end{align*}
		
		We proceed with the three proofs of monotonicity, which we show for the functionals
		\begin{equation}
			F_{\lambda,\vec{\phi}}(\Theta):= \max_{\Phi} \left|\left|\Delta \Theta \Phi \left(\lambda-\mu \Lambda_{\vec{\phi}}\right)\right|\right|_1.\label{eq:F_Lgen}
		\end{equation} The extension to $M_{\lambda,\vec{\phi}}(\Theta)$ is straightforward since the functionals only differ by a constant.
		
		We start proving that the functionals $F_{\lambda,\vec{\phi}}(\Theta)$ satisfy Eq.~(\ref{eq:sx}) by exploiting that a detection incoherent channel $\Theta_{\mathcal{DI}}$ cannot turn an incoherent POVM into a coherent one~\cite{Meznaric2013}. We therefore find 
		\begin{align*}
			{F}_{\lambda,\vec{\phi}} &(\Theta_{\mathcal{DI}} \Theta) = \max_{\Phi,\rho}\Tr\left|\Delta  \Theta_{\mathcal{DI}} \Theta \Phi  \left(\lambda-\mu \Lambda_{\vec{\phi}}\right)(\rho)\right| \\
			&= \max_{\Phi,\rho}\max_{P_0\in\mathcal{P}_I}\Tr\left(P_0 \Theta_{\mathcal{DI}} \Theta \Phi  \left(\lambda-\mu \Lambda_{\vec{\phi}}\right)(\rho)\right) \\
			&= \max_{\Phi,\rho}\max_{P_0\in\mathcal{P}_I,P_0'=P_0\Theta_{\mathcal{DI}}}\Tr\left(P_0'\Theta \Phi  \left(\lambda-\mu \Lambda_{\vec{\phi}}\right)(\rho)\right)\\
			&\leq \max_{\Phi,\rho}\max_{P_0'\in\mathcal{P}_I}\Tr\left(P_0'\Theta \Phi  \left(\lambda-\mu \Lambda_{\vec{\phi}}\right)(\rho)\right) = {F}_{\lambda,\vec{\phi}}(\Theta),
		\end{align*}
		where the inequality is due to an extension of the set over which we maximize.
		
		That the $F_{\lambda,\vec{\phi}}(\Theta)$ also satisfy Eq.~(\ref{eq:dx}) is proven by
		\begin{align*}
			{F}_{\lambda,\vec{\phi}}(\Theta &\Theta_{\mathcal{DI}}) =\max_{\Phi\in\mathcal{DI},\rho}\Tr\left|\Delta  \Theta \Theta_{\mathcal{DI}} \Phi  \left(\lambda-\mu \Lambda_{\vec{\phi}}\right)(\rho)\right| \\
			&=\max_{\Phi\in\mathcal{DI},\,\Phi'= \Theta_{\mathcal{DI}}\Phi,\rho}\Tr\left|\Delta \Theta \Phi'  \left(\lambda-\mu \Lambda_{\vec{\phi}}\right)(\rho)\right| \\
			&\leq \max_{\Phi'\in\mathcal{DI},\rho}\Tr\left|\Delta  \Theta \Phi'  \left(\lambda-\mu \Lambda_{\vec{\phi}}\right)(\rho)\right| = {F}_{\lambda,\vec{\phi}}(\Theta).
		\end{align*}
		
		Next, we prove satisfaction of Eq.~(\ref{eq:ce}). Indeed, we will prove the slightly more general statement that the functionals $F_{\lambda,\vec{\phi}}(\Theta)$ are constant under tensor product with the identity. We begin with
		\small
		\begin{align*}
			F_{\lambda,\vec{\phi}} &\left(\Theta^{A}\otimes \id^B\right) =\max_{\Phi^{AB\leftarrow C}}  \bigg| \bigg| \Delta \left(\Theta^{A}\otimes \id^B\right) \\
			&\pushright{\Phi^{AB\leftarrow C} \left(\lambda-\mu \Lambda_{\vec{\phi}}\right)\bigg|\bigg|_1} \\
			&=\max_{\Phi^{AB\leftarrow C}}  \bigg| \bigg| \left(\Delta^A \Theta^{A}\otimes\id^B\right) \left(\id^{A}\otimes\Delta^B\right) \\
			&\pushright{ \Phi^{AB\leftarrow C} \left(\lambda-\mu \Lambda_{\vec{\phi}}\right)\bigg|\bigg|_1} \\
			&\leq\max_{p_b,\tilde{\Phi}^{A\leftarrow C}_{b},\rho}  \Bigg| \Bigg| \left(\Delta^A \Theta^{A}\otimes\id^B\right) \\
			&\pushright{\sum_bp_b\left(\tilde{\Phi}^{A\leftarrow C}_{b} \left(\lambda-\mu\Lambda_{\vec{\phi}}\right)(\rho)\right)\otimes\ketbra{b}{b}_B\Bigg|\Bigg|_1} \\
			&\leq \max_{p_b,\tilde{\Phi}^{A\leftarrow C}_{b},\rho} \sum_b p_b \bigg| \bigg| \left(\Delta^A \Theta^{A}\otimes\id^B\right) \\
			&\pushright{\left(\tilde{\Phi}^{A\leftarrow C}_{b} \left(\lambda-\mu\Lambda_{\vec{\phi}}\right)(\rho)\right)\otimes\ketbra{b}{b}_B\bigg|\bigg|_1} \\
			&= \max_{p_b,\tilde{\Phi}^{A\leftarrow C}_b,\rho} \sum_bp_b\left| \left| \Delta^A \Theta^{A}\left(\tilde{\Phi}^{A\leftarrow C}_{b} \left(\lambda-\mu\Lambda_{\vec{\phi}}\right)(\rho)\right)\right|\right|_1 \\
			&= \max_{\tilde{\Phi}^{A\leftarrow C}} \left| \left| \Delta^A \Theta^{A}\left(\tilde{\Phi}^{A\leftarrow C} \left(\lambda-\mu\Lambda_{\vec{\phi}}\right)\right)\right|\right|_1= F_{\lambda,\vec{\phi}}\left(\Theta^{A}\right),
		\end{align*}
		\normalsize
		where the first inequality follows from Lem.~\ref{lem:tensor} and the second from convexity of the trace norm. 
		
		On the other hand, we can also prove the inverse inequality. We have
		\small
		\begin{align*}
			F_{\lambda,\vec{\phi}}&\left(\Theta^{A}\otimes \id^B\right) =\max_{\Phi^{AB\leftarrow C}}  \bigg| \bigg| \Delta \left(\Theta^{A}\otimes \id^B\right) \\
			&\pushright{\Phi^{AB\leftarrow C} \left(\lambda-\mu \Lambda_{\vec{\phi}}^C\right)\bigg|\bigg|_1} \\
			&\geq \max_{\Phi^{A\leftarrow C},\rho_B,\rho} \bigg| \bigg| \Delta \left(\Theta^{A}\otimes \id^B\right) \\
			&\pushright{\left(\Phi^{A\leftarrow C}\otimes\rho_B\right) \left( \left(\lambda-\mu \Lambda_{\vec{\phi}}^C\right)(\rho)\right)\bigg|\bigg|_1}\\
			&= \max_{\Phi^{A\leftarrow C},\rho,\rho_B} \left| \left| \Delta \Theta^{A} \Phi^{A\leftarrow C}  \left(\lambda-\mu \Lambda_{\vec{\phi}}^C\right)(\rho)\otimes\Delta^B(\rho_B)\right|\right|_1\\
			&= \max_{\Phi_A,\rho} \left| \left| \Delta \Theta^{A} \Phi^{A\leftarrow C}  \left(\lambda-\mu \Lambda_{\vec{\phi}}^C\right)(\rho)\right|\right|_1=F_{\lambda,\vec{\phi}}\left(\Theta^{A}\right),
		\end{align*}
		\normalsize
		where the inequality is due to a restriction of the set over which we maximize. We conclude that
		\begin{align*}
			M_{\lambda,\vec{\phi}}(\Theta\otimes\id^B)= M_{\lambda,\vec{\phi}}(\Theta)
		\end{align*} 
		for all $B$, $\Theta$, $\lambda$, and $\vec{{\phi}}$. 
		Finally, convexity follows straightforwardly from convexity of the trace norm. 
	\end{proof}

	\begin{theorem}
		An auxiliary system does not increase the pre-processed improvements, i.e., for $\vec{\tilde{\phi}}(AZ)$ such that 
		\begin{equation*}
			\Lambda_{\vec{\tilde{\phi}}(AZ)}:=\Lambda_{\vec{\phi}}^{A}\otimes\id^{Z},
		\end{equation*}
		it holds that
		\begin{equation*}
			M_{\lambda,\vec{\tilde{\phi}}(AZ)}(\Theta)= M_{\lambda,\vec{\phi}}(\Theta) \quad \forall \Theta\ \forall \lambda \ \forall {\vec{\phi}} \ \forall Z .
		\end{equation*}
	\end{theorem}
	
	\begin{proof} 
		Our proof rests upon showing an inequality in both directions. Also here, we will prove the inequalities for $F_{\lambda,\vec{\tilde{\phi}}(AZ)} (\Theta)$ (defined in Eq.~(\ref{eq:F_Lgen})), with the extension to $M_{\lambda,\vec{\tilde{\phi}}(AZ)} (\Theta)$ being straightforward. 
		
		In one direction, we have
		\begin{align*}
			&F_{\lambda,\vec{\tilde{\phi}}(AZ)} (\Theta^{C\leftarrow B})=F_{\lambda,\vec{\tilde{\phi}}(AZ)} (\Theta^{C\leftarrow B}\otimes\id^D) \\
			&= \max_{\Phi^{BD\leftarrow AZ}} \bigg| \bigg| \Delta \left(\Theta^{C\leftarrow B}\otimes\id^D\right) \Phi^{BD\leftarrow AZ}  \\
			&\pushright{\left(\left(\lambda-\mu \Lambda_{\vec{\phi}}^A\right)\otimes\id^{Z}\right)\bigg|\bigg|_1} \\
			&\geq \max_{\Phi^{B\leftarrow A},\Phi^{D\leftarrow Z},\rho_A,\rho_Z} \bigg| \bigg| \Delta \left(\Theta^{C\leftarrow B}\otimes\id^D\right) \\
			&\pushright{\left(\Phi^{B\leftarrow A}\otimes\Phi^{D\leftarrow Z}\right)  \left(\left(\lambda-\mu \Lambda_{\vec{\phi}}^A\right)\otimes\id^{Z}\right)(\rho_A\otimes\rho_Z)\bigg|\bigg|_1} \\
			&= \max_{\Phi^{B\leftarrow A},\rho_A} \left| \left| \Delta \Theta^{C\leftarrow B} \Phi^{B\leftarrow A}  \left(\lambda-\mu \Lambda_{\vec{\phi}}^A\right)(\rho_A)\right|\right|_1  \\
			&\pushright{ \max_{\Phi^{D\leftarrow Z},\rho_Z} \left| \left| \Delta \id^D \Phi^{D\leftarrow Z} \id^{Z}(\rho_Z)\right|\right|_1 }\\
			&= \max_{\Phi^{B\leftarrow A},\rho_A} \left| \left| \Delta \Theta^{C\leftarrow B} \Phi^{B\leftarrow A}  \left(\lambda-\mu \Lambda_{\vec{\phi}}^A\right)(\rho_A)\right|\right|_1 \\
			&=F_{{\lambda},{\vec{\phi}}} (\Theta^{C\leftarrow B}),
		\end{align*}
		where the first equality is due to constancy under tensor product and the inequality is due to a restriction of the set of states and pre-processing operations over which we maximize.
		
		Now, we show the inverse direction,
		\begin{align*}
			F_{{\lambda},{\vec{\phi}}} &(\Theta^{C\leftarrow B})=\max_{\Phi^{B\leftarrow A}} \left| \left| \Delta \Theta^{C\leftarrow B} \Phi^{B\leftarrow A}  \left(\lambda-\mu \Lambda_{\vec{\phi}}^A\right)\right|\right|_1 \\
			&\geq \max_{\Phi^{B\leftarrow A}=\Phi^{B\leftarrow AZ} \Phi^{AZ\leftarrow A}} \bigg| \bigg| \Delta \Theta^{C\leftarrow B} \Phi^{B\leftarrow AZ} \\
			&\pushright{\Phi^{AZ\leftarrow A}  \left(\lambda-\mu \Lambda_{\vec{\phi}}^A\right)\bigg|\bigg|_1} \\
			&= \max_{\Phi^{B\leftarrow A}=\Phi^{B\leftarrow AZ} \Phi^{AZ\leftarrow A},\ket{\varphi}_A} \bigg| \bigg| \Delta \Theta^{C\leftarrow B} \Phi^{B\leftarrow AZ} \\
			&\pushright{\Phi^{AZ\leftarrow A}  \left(\lambda-\mu \Lambda_{\vec{\phi}}^A\right)(\ketbra{\varphi}{\varphi}_A)\bigg|\bigg|_1} \\
			&\geq \max_{\Phi^{B\leftarrow AZ},\Phi^{AZ\leftarrow A} \ \land \ \ket{\varphi}_A \ \text{: Eq.~(\ref{eq:isometry}) is satisfied}} \bigg| \bigg| \Delta \Theta^{C\leftarrow B} \\
			&\pushright{\Phi^{B\leftarrow AZ} \Phi^{AZ\leftarrow A}  \left(\lambda-\mu \Lambda_{\vec{\phi}}^A\right)(\ketbra{\varphi}{\varphi}_A)\bigg|\bigg|_1} \\
			&=\max_{\Phi^{B\leftarrow AZ},\ket{\psi}_{AZ}} \bigg| \bigg| \Delta_{C} \Theta^{C\leftarrow B} \Phi^{B\leftarrow AZ}  \\
			&\pushright{\left(\left(\lambda-\mu \Lambda_{\vec{\phi}}^A\right)\otimes\id^{Z}\right)(\ketbra{\psi}{\psi}_{AZ})\bigg|\bigg|_1 }\\
			&=F_{\lambda,\vec{\tilde{\phi}}(AZ)} (\Theta^{C\leftarrow B}),
		\end{align*}
		where the first inequality is due to a restriction of the set over which we  maximize. In the following equality, we used that the maximum is always achieved on pure states due to convexity of the trace norm. The second inequality follows from a further restriction of the set over which we maximize. 
		The  equality thereafter is due to Lem.~\ref{lem:free_exp}.
	\end{proof}

	\begin{theorem}
		For all $\vec{\phi}\in\mathbbm{R}^M$ with at least two different components, the functionals $M_{\lambda,\vec{\phi}}(\Theta)$ are faithful if and only if $\lambda=\frac{1}{2}$.
	\end{theorem}
	\begin{proof}
		For this proof, we define the functionals
		\begin{equation}
			F_{\lambda,\vec{\phi}}(\Theta,\Phi,\rho):= \left|\left|\Delta \Theta \Phi \left(\lambda-\mu \Lambda_{\vec{\phi}}\right)(\rho)\right|\right|_1. \label{eq:F_L-nomax}
		\end{equation}
		
		Let us start with $\lambda=\mu=\frac{1}{2}$. In this case, we have $M_{\frac{1}{2},\vec{\phi}}(\Theta)=F_{\frac{1}{2},\vec{\phi}}(\Theta)$. 
		
		Now assume that $\Theta$ can detect coherence. In other words, we are assuming that $\exists \,\tilde{m},\tilde{k},\tilde{l}$ such that $\tilde{k}\neq \tilde{l}$ and $\Theta^{\tilde{k},\tilde{l}}_{\tilde{m},\tilde{m}}:=\bra{\tilde{m}}\Theta(\ketbra{\tilde{k}}{\tilde{l}})\ket{\tilde{m}}\neq 0$. For simplicity, we take $\tilde{m}=0$ since we can always relabel the indices. We will show that in this case, there exists a choice of $\Phi$ and $\rho$ such that $F_{\frac{1}{2},\vec{\phi}}(\Theta,\Phi,\rho)>0$, for $F_{\frac{1}{2},\vec{\phi}}(\Theta,\Phi,\rho)$ as defined in Eq.~(\ref{eq:F_L-nomax}). This implies that for non-free $\Theta$, we find $M_{\frac{1}{2},\vec{\phi}}(\Theta)>0$, which proves faithfulness. 
		
		First, by assumption, there exists a couple $\tilde{i}\ne\tilde{j}$ such that $\phi:=\phi_{\tilde{i}}-\phi_{\tilde{j}}\neq 0$. Writing again $\rho=\sum_{i,j} \rho_{i,j}\ketbra{i}{j}$, we have
		\begin{equation}
			\left(\lambda-\mu\Lambda_{\vec{\phi}}\right)(\rho)=\sum_{i,j}\ketbra{i}{j}\left(\lambda-\mu e^{i(\phi_i-\phi_j)}\right)\rho_{i,j}. \label{eq:lamu}
		\end{equation}
		Now we choose our initial state $\tilde{\rho}$ such that $\rho_{\tilde{i},\tilde{i}}=\rho_{\tilde{j},\tilde{j}}=\frac{1}{2}$ and $\rho_{\tilde{i},\tilde{j}}=\frac{1}{2}e^{i\xi}=\rho_{\tilde{j},\tilde{i}}^*$. Thus
		\begin{align*}
			\bigg(\lambda &-\mu\Lambda_{\vec{\phi}}\bigg)(\tilde{\rho})
			=\frac{1}{2}\left(\ketbra{\tilde{i}}{\tilde{i}}+\ketbra{\tilde{j}}{\tilde{j}}\right)(\lambda-\mu) \\
			&\pushright{+\frac{1}{2}\left(\ketbra{\tilde{i}}{\tilde{j}}e^{i\xi}\left(\lambda-\mu e^{i\phi}\right)+\ketbra{\tilde{j}}{\tilde{i}}e^{-i\xi}\left(\lambda-\mu e^{-i\phi}\right)\right).}
		\end{align*}
		The free operation $\Phi$ we choose is  $\Phi_{\text{SWAP}}$ defined via the unitary
		\begin{align*}
			U_{\text{SWAP}}=&\ketbra{\tilde{i}}{\tilde{k}}+\ketbra{\tilde{k}}{\tilde{i}}+\ketbra{\tilde{j}}{\tilde{l}}+\ketbra{\tilde{l}}{\tilde{j}}+\mathbbm{1} \\
			&\pushright{-\ketbra{\tilde{i}}{\tilde{i}}-\ketbra{\tilde{k}}{\tilde{k}}-\ketbra{\tilde{j}}{\tilde{j}}-\ketbra{\tilde{l}}{\tilde{l}}.}
		\end{align*} 
		According to our assumptions, $\Theta$ detects relative phases between $\tilde{k}$, $\tilde{l}$ and $\Lambda_{\vec{\phi}}$ encodes a relative phase between $\tilde{i}$, $\tilde{j}$. The operation $\Phi_{\text{SWAP}}$ swaps the qubits that these pairs of indices define and leaves the remainder of the space unchanged, which ensures that the subspace in which $\Theta$ detects is aligned with the one in which $\Lambda_{\vec{\phi}}$ encodes (see also App.~\ref{ap:preproc} why this might be necessary).
		This will now allow us to prove faithfulness. From
		\begin{align*}
			\Phi_{\text{SWAP}}\!&\left(\lambda-\mu\Lambda_{\vec{\phi}}\right)(\tilde{\rho}) =U_{\text{SWAP}} \left(\lambda-\mu\Lambda_{\vec{\phi}}\right)(\tilde{\rho})U_{\text{SWAP}}^{\dagger} \\
			=&\frac{1}{2}\bigg(\left(\ketbra{\tilde{k}}{\tilde{k}}+ \ketbra{\tilde{l}}{\tilde{l}}\right)(\lambda-\mu) \\
			& +\ketbra{\tilde{k}}{\tilde{l}}e^{i\xi}\left(\lambda-\mu e^{i\phi}\right)
			+\ketbra{\tilde{l}}{\tilde{k}}e^{-i\xi}\left(\lambda-\mu e^{-i\phi}\right)\bigg),
		\end{align*}
		and inserting explicit $\lambda=\frac{1}{2}$, follows
		\begin{align*}
			\Delta \Theta & \Phi_{\text{SWAP}} \left(\lambda-\mu\Lambda_{\vec{\phi}}\right)(\tilde{\rho})
			=\frac{1}{4}\sum_m\bigg(\Theta^{\tilde{k},\tilde{l}}_{m,m}\ketbra{m}{m}e^{i\xi} \\
			&\pushright{\left(1- e^{i\phi}\right) +\Theta^{\tilde{l},\tilde{k}}_{m,m}\ketbra{m}{m}e^{-i\xi}\left(1- e^{-i\phi}\right) \bigg)}\\
			&=\frac{1}{2}\sum_m\ketbra{m}{m}\Re\left(\Theta^{\tilde{k},\tilde{l}}_{m,m}e^{i\xi}\left(1- e^{i\phi}\right)\right),
		\end{align*}
		where we used the index representation introduced in App.~\ref{ap:index}.
		Now, we have
		\begin{align*}
			F_{\frac{1}{2},\vec{\phi}}(\Theta,\Phi_{\text{SWAP}},\tilde{\rho})
			=&\frac{1}{2}\sum_{m=0}^{M-1}\left|\Re\left(\Theta^{\tilde{k},\tilde{l}}_{m,m}e^{i\xi}\left(1- e^{i\phi}\right)\right)\right| \\
			=&\frac{1}{2}\left|\Re\left(\Theta^{\tilde{k},\tilde{l}}_{0,0}e^{i\xi}\left(1- e^{i\phi}\right)\right)\right| \\
			& +\frac{1}{2}\sum_{m=1}^{M-1}\left|\Re\left(\Theta^{\tilde{k},\tilde{l}}_{m,m}e^{i\xi}\left(1- e^{i\phi}\right)\right)\right| \\
			\geq& \frac{1}{2}\left|\Re\left(\Theta^{\tilde{k},\tilde{l}}_{0,0}e^{i\xi}\left(1- e^{i\phi}\right)\right)\right| \\
			& +\frac{1}{2}\left|\Re\left(\sum_{m=1}^{M-1}\Theta^{\tilde{k},\tilde{l}}_{m,m}e^{i\xi}\left(1- e^{i\phi}\right)\right)\right| \\
			=&\left|\Re\left(\Theta^{\tilde{k},\tilde{l}}_{0,0}e^{i\xi}\left(1- e^{i\phi}\right)\right)\right|,
		\end{align*}
		where the inequality stems from the triangle inequality and the last line from Proposition~\ref{prop:ind_pr}.3 in App.~\ref{ap:index}.
		
		Finally, recalling that $\max_{\xi}\Re(Ae^{i\xi})=|A|$, we conclude that
		\begin{align*}
			M_{\frac{1}{2},\vec{\phi}}(\Theta)&\geq \max_{\xi}F_{\frac{1}{2},\vec{\phi}}(\Theta,\Phi_{\text{SWAP}},\tilde{\rho}) \\
			& \geq \left|\Theta^{\tilde{k},\tilde{l}}_{0,0}\left(1- e^{i\phi}\right)\right|=\left|\Theta^{\tilde{k},\tilde{l}}_{0,0}\right| \left|1- e^{i\phi}\right|>0
		\end{align*}
		because by assumption $\Theta^{\tilde{k},\tilde{l}}_{0,0}\neq 0$ and $\phi\neq 0$. Together with Thm.~\ref{theo:meas_pre}, this finishes the case $\lambda=\mu=\frac{1}{2}$.
		
		Next, we show that whenever $\lambda\neq\mu$, there exists a channel $\tilde{\Theta}\notin \mathcal{DI}$ such that $M_{\lambda,\vec{\phi}}(\tilde{\Theta})=0$ for all $\vec{\phi}$. 
		With $\Phi\in \mathcal{DI}$ and starting from Eq.~(\ref{eq:lamu}), we have
		\begin{align}
			\Phi &\left(\lambda-\mu\Lambda_{\vec{\phi}}\right)(\rho)=\sum_{i,k}\ketbra{k}{k}\Phi^{i,i}_{k,k}\rho_{i,i}(\lambda-\mu ) \nonumber \\
			&\pushright{+\sum_{i,j,k\neq l}\ketbra{k}{l}\Phi^{i,j}_{k,l}\rho_{i,j}\left(\lambda-\mu e^{i(\phi_i-\phi_j)}\right),}
		\end{align} 
		where we made use of Eq.~\eqref{eq:indexDetIn} of App.~\ref{ap:index}.
		Using this representation, we get
		\begin{align}\label{eq:F_L}
			F_{\lambda,\vec{\phi}}&(\Theta,\Phi,\rho)=\sum_m\Bigg|\sum_{i,k}\Theta^{k,k}_{m,m}\Phi^{i,i}_{k,k}\rho_{i,i}(\lambda-\mu ) \nonumber \\
			&\pushright{+\sum_{i,j,k\neq l}\Theta^{k,l}_{m,m}\Phi^{i,j}_{k,l}\rho_{i,j}\left(\lambda-\mu e^{i(\phi_i-\phi_j)}\right) \Bigg| \nonumber} \\
			=&\sum_m\Bigg|\sum_{i,k}\Theta^{k,k}_{{m},{m}}\Phi^{i,i}_{k,k}\rho_{i,i}(\lambda-\mu )  \\
			&\pushright{+2\sum_{k>l,i,j}\Re\left(\Theta^{k,l}_{{m},{m}}\Phi^{i,j}_{k,l}\rho_{i,j}\left(\lambda-\mu e^{i(\phi_i-\phi_j)}\right) \right) \Bigg|.}\nonumber 
		\end{align}
		For a more convenient notation, we use the quantities
		\begin{equation*}
			A_m:= 2\sum_{k>l,i,j}\Re\left(\Theta^{k,l}_{{m},{m}}\Phi^{i,j}_{k,l}\rho_{i,j}\left(\lambda-\mu e^{i(\phi_i-\phi_j)}\right) \right)
		\end{equation*}
		from here on.
		In order for $M_{\lambda,\vec{\phi}}(\Theta)$ to be faithful, we need that $F_{\lambda,\vec{\phi}}(\Theta)>|\lambda-\mu | \ \forall\Theta\notin\mathcal{DI}$. 
		We further notice that, if the terms inside the absolute value in Eq.~(\ref{eq:F_L}) have  the same sign for all $m$, then $F_{\lambda,\vec{\phi}}(\Theta)=|\lambda-\mu |$ due to trace preservation. Hence, a necessary condition for $M_{\lambda,\vec{\phi}}(\Theta)>0$ is that there exists an $\tilde{m}$ such that
		\begin{equation*}
			\sum_{i,k}\Theta^{k,k}_{\tilde{m},\tilde{m}}\Phi^{i,i}_{k,k}\rho_{i,i}(\lambda-\mu )+A_m<0.
		\end{equation*}
		From here on, we consider a specific detecting operation that will violate this condition, namely
		\begin{equation}\label{eq:mixop}
			\tilde{\Theta}(\rho)=p_1 Q\rho Q^\dagger+p_2\rho,
		\end{equation} where $p_1+p_2=1$, $p_i\ge0$, and  $Q$ a costly unitary (e.g., a quantum Fourier transform). As long as $p_1>0$, $\tilde{\Theta}$ is thus non-free. With this choice, we find that
		\begin{equation*}
			\tilde{\Theta}^{k,k}_{m,m}=p_1 |Q_{m,k}|^2+p_2\delta_{m,k},
		\end{equation*}
		\begin{equation*}
			\tilde{\Theta}^{k,l}_{m,m}=p_1Q_{m,k}Q_{m,l}^* \quad\forall k\neq l,
		\end{equation*} where $Q_{m,k}$ are the matrix elements of the unitary $Q$. 
		For faithfulness to hold, we thus need that there exists an $\tilde{m}$ such that 
		\begin{equation*}
			\sum_{i,k}(p_1 |Q_{\tilde{m},k}|^2+p_2\delta_{\tilde{m},k})\Phi^{i,i}_{k,k}\rho_{i,i}(\lambda-\mu )+A_{\tilde{m}}<0,
		\end{equation*} 
		or, using $p_2=1-p_1$, that
		\small
		\begin{equation*}
			p_1\left((\lambda-\mu)\left(\sum_{i,k}|Q_{\tilde{m},k}|^2\Phi^{i,i}_{k,k}\rho_{i,i}-\sum_i\Phi^{i,i}_{\tilde{m},\tilde{m}}\rho_{i,i}\right)+A_{\tilde{m}} \right) 
		\end{equation*}
		\normalsize
		\begin{equation}
			<-(\lambda-\mu)\sum_i\Phi^{i,i}_{\tilde{m},\tilde{m}}\rho_{i,i}.\label{eq:p1}
		\end{equation}
		At this stage, we make some further assumptions. Firstly, we only consider the case $\lambda>\mu$ (the other case is analogous). Secondly, we choose the detecting channel such that its input dimension is two. This forces the output dimension of the pre-processing to be two as well. 
		
		To finish the proof, we now examine two different cases, beginning with the one in which $\exists\ m' \ \text{s.t.} \ \forall i \ \Phi^{i,i}_{m',m'}=0$. Then, using Lem.~\ref{lem:detect_two}, we obtain $\Phi^{i,j}_{k,l}=0 \ \forall k\neq l, \ \forall i,j$. This means that Eq.~(\ref{eq:p1}), assuming w.l.o.g. $\tilde{m}=0$, becomes 
		\begin{align*}
			&p_1(\lambda-\mu)|Q_{0,1}|^2<0 \qquad &\text{if} \ m'=0, \\
			&p_1(\lambda-\mu)|Q_{0,0}|^2<(\lambda-\mu)(p_1-1) \qquad &\text{if} \ m'=1,
		\end{align*} where we used the fact that the trace of a density operator is one. One easily verifies that the inequalities are independent of $\rho$ and $\Phi$ and cannot be satisfied in either case. 
		
		The second case is the one in which $\forall\ m \ \exists i \ \text{s.t.} \ \Phi^{i,i}_{m,m}\neq 0$. Recalling Proposition~\ref{prop:ind_pr}.1 of App.~\ref{ap:index}, we thus deduce that the right hand side of Eq.~(\ref{eq:p1}) is negative. We next choose $p_1>0$ such that
		\small
		\begin{equation*}
			p_1\leq\min_m\frac{(\lambda-\mu)\sum_i\Phi^{i,i}_{{m},{m}}\rho_{i,i}}{\left|\!(\lambda\!-\!\mu)\!\left(\sum_{i,k}|Q_{{m},k}|^2\Phi^{i,i}_{k,k}\rho_{i,i}\!-\!\sum_{i}\Phi^{i,i}_{{m},{m}}\rho_{i,i}\right)\!+\!A_m \!\right|} .
		\end{equation*}
		\normalsize
		Note that this is always possible, because all quantities in the above fraction are finite and the numerator cannot be zero by assumption. Moreover, when the denominator is zero, we choose an arbitrary $p_1$ with $0< p_1\le1$. For this choice of $p_1$, the inequality (Eq.~(\ref{eq:p1})) cannot be satisfied. 
		
		We conclude that, if $\lambda\neq\mu$, there exist non-free operations for which $M_{\lambda,\vec{\phi}}(\Theta)=0$.
	\end{proof}

	\begin{theorem}
		The functionals $N_{\lambda,\vec{\phi}}(\Theta)$ are convex measures in the creation incoherent setting for all $\lambda\in [0,1]$ and for all $\vec{\phi}\in\mathbbm{R}^M$.
	\end{theorem}
	
	\begin{proof} Large parts of this proof are very similar to the proof of Thm.~\ref{theo:meas_pre}. For completeness, we give the full proof nevertheless, starting with nullity. Let us assume that $\Theta\in\mathcal{MIO}$, i.e., $\Theta \Delta=\Delta \Theta \Delta$. If $\Psi$ is also creation incoherent, then $\Psi\Theta \Delta=\Delta \Psi\Theta \Delta$.  
		Noticing that $\Lambda_{\vec{\phi}} \Delta=\Delta$, we thus get
		\begin{align*}
			N_{\lambda,\vec{\phi}} &(\Theta)=\max_{\Psi\in\mathcal{MIO}}\left|\left|\left(\lambda-\mu\Lambda_{\vec{\phi}}\right) \Psi \Theta\Delta\right|\right|_1-|\lambda-\mu | \\
			&=\max_{\Psi\in\mathcal{MIO}}\left|\left|\left(\lambda-\mu\Lambda_{\vec{\phi}}\right) \Delta \Psi \Theta \Delta\right|\right|_1-|\lambda-\mu | \\
			&=\max_{\Psi\in\mathcal{MIO}}\left|\left|\left(\lambda-\mu\right)\Delta \Psi \Theta \Delta\right|\right|_1-|\lambda-\mu | \\
			&=|\lambda-\mu |\max_{\Psi\in\mathcal{MIO}}\left|\left|\Delta \Psi \Theta \Delta\right|\right|_1-|\lambda-\mu | \\
			&=|\lambda-\mu |-|\lambda-\mu |=0.
		\end{align*}
		Non-negativity holds due to
		\begin{align*}
			N_{\lambda,\vec{\phi}}&(\Theta)=\max_{\Psi\in\mathcal{MIO},\rho\in\mathcal{I}}\Tr\left|\left(\lambda-\mu\Lambda_{\vec{\phi}}\right) \Psi \Theta(\rho)\right|-|\lambda-\mu | \\
			&\geq\max_{\Psi\in\mathcal{MIO},\rho\in\mathcal{I}}\left|\Tr\left(\left(\lambda-\mu\Lambda_{\vec{\phi}}\right) \Psi \Theta(\rho)\right)\right|-|\lambda-\mu | \\
			&=|\lambda-\mu |-|\lambda-\mu |=0.
		\end{align*}
		
		We now proceed with the three proofs of monotonicity, again for the functionals
		\begin{equation}
			G_{\lambda,{\vec{\phi}}}(\Theta):=\max_{\Psi\in\mathcal{MIO}}\left|\left|\left(\lambda-\mu \Lambda_{\vec{\phi}}\right) \Psi \Theta\Delta\right|\right|_1, \label{eq:G_L}
		\end{equation} since the extensions to $N_{\lambda,{\vec{\phi}}}(\Theta)$ are straightforward.
		
		We begin showing that all $G_{\lambda,{\vec{\phi}}}(\Theta)$ satisfy Eq.~(\ref{eq:sx}). Let us assume $\Theta_0\in\mathcal{MIO}$. Then
		\begin{align*}
			G_{\lambda,\vec{\phi}}&(\Theta_0 \Theta)=\max_{\Psi\in\mathcal{MIO},\rho\in\mathcal{I}}\Tr\left|\left(\lambda-\mu\Lambda_{\vec{\phi}}\right) \Psi \Theta_0 \Theta(\rho)\right| \\
			&=\max_{\Psi'=\Psi \Theta_0\in\mathcal{MIO},\rho\in\mathcal{I}}\Tr\left|\left(\lambda-\mu\Lambda_{\vec{\phi}}\right) \Psi' \Theta(\rho)\right| \\
			&\leq\max_{\Psi'\in\mathcal{MIO},\rho\in\mathcal{I}}\Tr\left|\left(\lambda-\mu\Lambda_{\vec{\phi}}\right) \Psi' \Theta(\rho)\right|=G_{\lambda,\vec{\phi}}(\Theta).
		\end{align*}
		Eq.~(\ref{eq:dx}) is always satisfied too because, with $\Theta_0\in\mathcal{MIO}$ again, we have
		\begin{align*}
			G_{\lambda,\vec{\phi}}&(\Theta \Theta_0)=\max_{\Psi\in\mathcal{MIO},\rho\in\mathcal{I}}\Tr\left|\left(\lambda-\mu\Lambda_{\vec{\phi}}\right) \Psi \Theta \Theta_0(\rho)\right| \\
			&=\max_{\Psi\in\mathcal{MIO},\rho'=\Theta_0(\rho)\in\mathcal{I}}\Tr\left|\left(\lambda-\mu\Lambda_{\vec{\phi}}\right) \Psi \Theta\left(\rho'\right)\right| \\
			&\leq\max_{\Psi\in\mathcal{MIO},\rho'\in\mathcal{I}}\Tr\left|\left(\lambda-\mu\Lambda_{\vec{\phi}}\right) \Psi \Theta(\rho')\right|=G_{\lambda,\vec{\phi}}(\Theta).
		\end{align*}
		To conclude, we need to prove validity of Eq.~(\ref{eq:ce}). Again, we will show that the $G_{\lambda,\vec{\phi}}(\Theta)$ are constant under tensor product by establishing inequalities in both directions. We begin with
		\begin{align*}
			&G_{\lambda,\vec{\phi}}\left(\Theta^{B\leftarrow A}\otimes\id^Z\right)\\
			&=\max_{\Psi^{C\leftarrow BZ}\in\mathcal{MIO},\rho_{AZ}\in\mathcal{I}}\bigg|\bigg|\left(\lambda-\mu\Lambda_{\vec{\phi}}\right) \\
			&\pushright{ \Psi^{C\leftarrow BZ} \left(\Theta^{B\leftarrow A}\otimes\id^Z\right)(\rho_{AZ})\bigg|\bigg|_1} \\
			&=\max_{\Psi^{C\leftarrow BZ}\in\mathcal{MIO},\rho_{AZ}\in\mathcal{I}}\bigg|\bigg|\left(\lambda-\mu\Lambda_{\vec{\phi}}\right) \\
			&\pushright{\sum_ip_i(\rho_{AZ})\Psi_i^{C\leftarrow B}(\Psi^{C\leftarrow BZ}) \Theta^{B\leftarrow A}\left(\rho_{|i}(\rho_{AZ})\right)\bigg|\bigg|_1} \\
			&\leq\max_{p_i,\Psi_i^{C\leftarrow B}\in\mathcal{MIO},\rho_{|i}\in\mathcal{I}}\bigg|\bigg|\left(\lambda-\mu\Lambda_{\vec{\phi}}\right) \\
			&\pushright{\sum_ip_i\Psi_i^{C\leftarrow B} \Theta^{B\leftarrow A}(\rho_{|i})\bigg|\bigg|_1} \\
			&\leq\max_{p_i,\Psi_i^{C\leftarrow B}\in\mathcal{MIO},\rho_{|i}\in\mathcal{I}}\sum_ip_i \\
			&\pushright{ \left|\left|\left(\lambda-\mu\Lambda_{\vec{\phi}}\right) \Psi_i^{C\leftarrow B} \Theta^{B\leftarrow A}(\rho_{|i})\right|\right|_1} \\
			&\leq\max_{\Psi^{C\leftarrow B}\in\mathcal{MIO},\rho\in\mathcal{I}}\left|\left|\left(\lambda-\mu\Lambda_{\vec{\phi}}\right) \Psi^{C\leftarrow B} \Theta^{B\leftarrow A}(\rho)\right|\right|_1 \\
			&=G_{\lambda,\vec{\phi}}(\Theta^{B\leftarrow A}),
		\end{align*} 
		where the second equality is to be understood in the sense of Lem.~\ref{lem:post_tensor}. The first inequality is due to an extension of the set over which we maximize and the second inequality follows from convexity of the trace norm.
		
		Finally, we prove the inverse inequality with
		\begin{align*}
			G_{\lambda,\vec{\phi}}(\Theta^{B\leftarrow A}\otimes\id^Z)&\geq G_{\lambda,\vec{\phi}}\left(\Tr_Z \left(\Theta^{B\leftarrow A}\otimes\id^Z\right)\right) \\
			&=G_{\lambda,\vec{\phi}}(\Theta^{B\leftarrow A}),
		\end{align*}
		where we used monotonicity under right composition, and conclude that the post-processed improvements are constant under tensor product with the identity channel.
		
		Also here, convexity follows straightforwardly from convexity of the trace norm. The post-processed improvements are thus convex measures in the detection incoherent setting.
	\end{proof}

	\begin{theorem}
		The functionals $N_{\frac{1}{2},\vec{\phi}}(\Theta)$ are faithful for all $\vec{\phi}\in\mathbbm{R}^M$ with at least two different components.
	\end{theorem}
	
	\begin{proof} The proof is similar to the corresponding part of the proof of Thm.~\ref{thm:faithfulnessPre}. By assumption,  $\Lambda_{\vec{\phi}}$ cannot add only a global phase, thus
		\begin{equation*}
			\exists \tilde{m},\tilde{n} \quad \text{s.t.} \quad \tilde{m}\neq\tilde{n} \ \text{and} \ \phi_{\tilde{m}}\neq\phi_{\tilde{n}}.
		\end{equation*}
		We further assume that the operation $\Theta:A\rightarrow B$ is not free, i.e., 
		\begin{equation*}
			\exists \tilde{i},\tilde{k},\tilde{l} \quad \text{s.t.} \quad \tilde{k}\neq\tilde{l} \ \text{and} \ \Theta^{\tilde{i},\tilde{i}}_{\tilde{k},\tilde{l}}\neq 0.
		\end{equation*}
		As in the proof of Thm.~\ref{thm:faithfulnessPre}, we now consider an explicit choice of an incoherent state, $\tilde{\rho}=\ketbra{\tilde{i}}{\tilde{i}}$, and a free operation $\tilde{\Psi}$
		that is composed of two creation incoherent channels. The first one is a  SWAP operation defined via the unitary
		\begin{align*}
			U_{\text{SWAP}}&=\ketbra{\tilde{m}}{\tilde{k}}+\ketbra{\tilde{k}}{\tilde{m}}+\ketbra{\tilde{n}}{\tilde{l}}+\ketbra{\tilde{l}}{\tilde{n}}+\mathbbm{1} \\
			&\quad -\ketbra{\tilde{m}}{\tilde{m}}-\ketbra{\tilde{k}}{\tilde{k}}-\ketbra{\tilde{n}}{\tilde{n}}-\ketbra{\tilde{l}}{\tilde{l}}
		\end{align*} 
		that aligns again the relevant subspaces (see the proof of Thm.~\ref{thm:faithfulnessPre})
		and the second channel is defined via Kraus operators
		\begin{equation*}
			K:=\ketbra{\tilde{m}}{\tilde{m}}+\ketbra{\tilde{n}}{\tilde{n}}, \qquad L_j:=\ketbra{j}{j}\quad\forall j \neq\tilde{m},\tilde{n}.
		\end{equation*}
		We therefore have
		\begin{equation*}
			\Theta(\tilde{\rho})=\sum_{k,l}\Theta^{\tilde{i},\tilde{i}}_{k,l}\ketbra{k}{l},
		\end{equation*}
		and
		\begin{align*}
			&\tilde{\Psi} \Theta(\tilde{\rho})\\
			&=\Theta^{\tilde{i},\tilde{i}}_{\tilde{k},\tilde{l}}\ketbra{\tilde{m}}{\tilde{n}}+\Theta^{\tilde{i},\tilde{i}}_{\tilde{l},\tilde{k}}\ketbra{\tilde{n}}{\tilde{m}}+\Theta^{\tilde{i},\tilde{i}}_{\tilde{k},\tilde{k}}\ketbra{\tilde{m}}{\tilde{m}}+\Theta^{\tilde{i},\tilde{i}}_{\tilde{l},\tilde{l}}\ketbra{\tilde{n}}{\tilde{n}} \\
			&\quad+\Theta^{\tilde{i},\tilde{i}}_{\tilde{m},\tilde{m}}\ketbra{\tilde{k}}{\tilde{k}}+\Theta^{\tilde{i},\tilde{i}}_{\tilde{n},\tilde{n}}\ketbra{\tilde{l}}{\tilde{l}}+\sum_{j\notin \{ \tilde{k},\tilde{l},\tilde{m},\tilde{n}\}}\Theta^{\tilde{i},\tilde{i}}_{j,j}\ketbra{j}{j}.
		\end{align*}
		This time, we define
		\begin{align*}
			A:=&\left(\id-\Lambda_{\vec{\phi}}\right) \tilde{\Psi} \Theta(\tilde{\rho}) \\
			=&\left(1-e^{i(\phi_{\tilde{m}}-\phi_{\tilde{n}})}\right)\Theta^{\tilde{i},\tilde{i}}_{k,l}\ketbra{\tilde{m}}{\tilde{n}} \\
			&\quad +\left(1-e^{-i(\phi_{\tilde{m}}-\phi_{\tilde{n}})}\right)\Theta^{\tilde{i},\tilde{i}}_{l,k}\ketbra{\tilde{n}}{\tilde{m}},
		\end{align*} 
		and have thus
		\begin{align*}
			A^\dagger A&=\left|\left(1-e^{i(\phi_{\tilde{m}}-\phi_{\tilde{n}})}\right)\Theta^{\tilde{i},\tilde{i}}_{k,l}\right|^2\ketbra{\tilde{m}}{\tilde{m}} \\
			&\quad+\left|\left(1-e^{i(\phi_{\tilde{m}}-\phi_{\tilde{n}})}\right)\Theta^{\tilde{i},\tilde{i}}_{k,l}\right|^2\ketbra{\tilde{n}}{\tilde{n}}
		\end{align*}
		and
		\begin{align*}
			\sqrt{A^\dagger A}&=\left|\left(1-e^{i(\phi_{\tilde{m}}-\phi_{\tilde{n}})}\right)\Theta^{\tilde{i},\tilde{i}}_{k,l}\right|\left(\ketbra{\tilde{m}}{\tilde{m}}+\ketbra{\tilde{n}}{\tilde{n}}\right).
		\end{align*}
		Now, since $\phi_{\tilde{m}}\neq\phi_{\tilde{n}}$ and $\Theta^{\tilde{i},\tilde{i}}_{\tilde{k},\tilde{l}}\neq 0$, we find
		\begin{align*}
			||A||_1=\Tr\sqrt{A^\dagger A}&=2\left|\left(1-e^{i(\phi_{\tilde{m}}-\phi_{\tilde{n}})}\right)\Theta^{\tilde{i},\tilde{i}}_{k,l}\right|>0
		\end{align*}
		and finally
		\begin{align*}
			N_{\frac{1}{2},\vec{\phi}}(\Theta)&=\frac{1}{2}\max_{\Psi\in\mathcal{MIO},\rho\in\mathcal{I}}\left|\left|\left(\id-\Lambda_{\vec{\phi}}\right) \Psi \Theta(\rho)\right|\right|_1\\
			&\ge\frac{1}{2}\left|\left|\left(\id-\Lambda_{\vec{\phi}}\right) \tilde{\Psi} \Theta(\tilde{\rho})\right|\right|_1=\frac{1}{2}||A||_1>0.
		\end{align*} 
		Together with Thm.~\ref{theo:meas_post}, this finishes the proof. 
	\end{proof}

	\section{The need for an optimal pre-processing}\label{ap:preproc}
	As mentioned in Sec.~\ref{sec:main}, in this Appendix, we show the necessity of the optimal pre-processing in the definitions of $M_{\lambda,\vec{\phi}}(\Theta)$ (see Eq.~\eqref{eq:pre_proc}).
	To this end, we assume that Bob cannot apply the optimal pre-processing prior to the operation $\Theta$. The analogues of $M_{\lambda,\vec{\phi}}(\Theta)$ are then 
	\begin{equation}
		L_{\lambda,\vec{\phi}}(\Theta):=\left|\left|\Delta \Theta \left(\lambda-\mu \Lambda_{\vec{\phi}}\right)\right|\right|_1-|\lambda-\mu|.
	\end{equation}
	Via the construction of an explicit counterexample, we now show that the $L_{\lambda,\vec{\phi}}(\Theta)$ are not monotonic in general. 
	
	To this end, we provide Bob with a detecting operation of the form $\tilde{\Theta}=\Theta^A\otimes\id^B$, where $\Theta^A\notin\mathcal{DI}$ and $\dim(A)=\dim(B)$. In addition, we choose $\lambda=\mu=\frac{1}{2}$ and $\Lambda_{\vec{\phi}}:=\id^A\otimes\Lambda_{\vec{\phi}'}^B$. The resulting guessing game for Alice and Bob is represented in Fig.~\ref{fig:notSWAP}, 
	\begin{figure}[ht]
		\includegraphics[width=.95\linewidth]{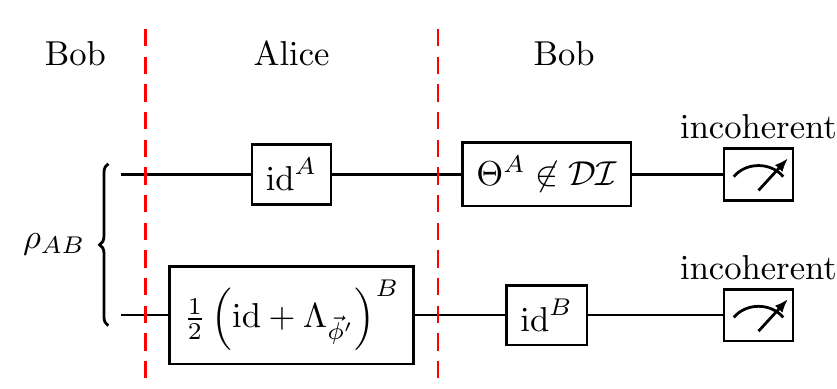}
		\caption{Schematic representation of the guessing game described in the main text.}\label{fig:notSWAP}
	\end{figure}
	and it is straightforward to verify that $L_{\vec{\phi}}(\tilde{\Theta})=0$, because $\tilde{\Theta}$ detects in subspace where $\Lambda_{\vec{\phi}}$ encodes no information. We now define a superchannel $S_0$ acting as 
	\begin{equation}
		S_0(\Theta)=\Theta \Theta_0, 
	\end{equation}
	where $\Theta_0$ is a SWAP channel that exchanges the systems $A$ and $B$, i.e., 
	\begin{equation}
		\Theta_{0}\left(\sum_{i,j,k,l}\rho_{ij,kl}\ketbra{ij}{kl}_{AB}\right)=\sum_{i,j,k,l}\rho_{ji,lk}\ketbra{ij}{kl}_{AB}.\label{eq:swapcounter}
	\end{equation} 
	The superchannel $S_0$ is free since $\Theta_0$ is free (it only relabels Hilbert spaces). Applying $S_0$ to $\tilde{\Theta}$, we obtain the situation represented in Fig.~\ref{fig:SWAP}, from where we deduce that $L_{\vec{\phi}}(S_0(\tilde{\Theta}))$ can be different from zero for specific choices of $\vec{\phi}'$ and $\Theta^A$, because now we detect in the subspace in which information is encoded. 
	\begin{figure}[ht]
		\includegraphics[width=.95\linewidth]{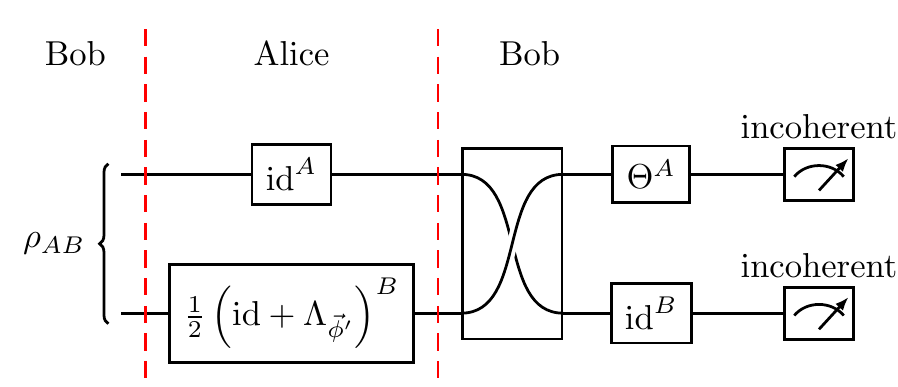}
		\caption{Addition of a SWAP operation to the game in Fig.~\ref{fig:notSWAP}, which shows that $L_{\lambda,\vec{\phi}}(\Theta)$ is not monotonic.}\label{fig:SWAP}
	\end{figure}
	We conclude that the functionals $L_{\lambda,\vec{\phi}}(\Theta)$ are in general not  measures in the detection incoherent setting.

	\section{Evaluation of the pre-processed improvements}\label{ap:eval}
	\setcounter{theorem}{10}
	Here we show how one can evaluate the pre-processed improvements numerically and how this leads to an optimal state and pre-processing. We begin with two Lemmas.
	
	\begin{lem}\label{lem:prelem1}
		The two sets 
		\begin{align*}
			\mathcal{X}:= &\Bigg\{ X_{AB}=\id^A\otimes \Theta^{B\leftarrow A} \sum_{i,j} \rho_{i,j} \ketbra{ii}{jj}_{AA} \Bigg| \\
			&\pushright{ \Theta \text{ CPTP, } \rho=\sum_{i,j}\rho_{i,j}\ketbra{i}{j} \text{ quantum state} \Bigg\},} \\
			\mathcal{Y}:= &\left\{Y_{AB}\middle \vert Y\ge0 ,\ \Tr_BY=\Delta (\sigma_A), \ \sigma_A  \text{ quantum state} \right\}
		\end{align*}
		are equal.	
	\end{lem}
	\begin{proof}
		First we show that every $X\in \mathcal{X}$ is also an element of $\mathcal{Y}$.
		Assume $X\in \mathcal{X}$. From this follows $X\ge0$ because we apply a quantum operation to a quantum state, leading to another quantum state. In addition
		\begin{align*}
			\Tr_B X = \sum_{i,j} \rho_{i,j} \ketbra{i}{j}_A \Tr\left( \Theta\left[ \ketbra{i}{j}_A\right] \right)=\sum_{i} \rho_{i,i}\ketbra{i}{i}_A
		\end{align*}
		is a diagonal state since $\rho$ is a state by assumption.
		
		To prove the reverse direction, we assume that $Y\in \mathcal{Y}$. From this assumption follows that $Y$ is a quantum state. Every purification $Z_{ABC}=\ketbra{\psi}{\psi}$ of $Y$, i.e., $\Tr_CZ=Y$, is by assumption also a purification of a $\Delta (\sigma_A)=\sum_i \sigma_i \ketbra{i}{i}_A$. Let us define the set $S=\left\{ i : \sigma_i\neq 0\right\}$ and $\tilde{A}=\vspan\left\{\ket{i}\in S\right\}$. Using the isometric freedom in purifications we find
		\begin{align*}
			\ket{\psi}=\sum_i \sqrt{\sigma_{i}}\ket{i}_A\otimes V^{BC\leftarrow \tilde{A}}\ket{i}_{\tilde{A}}
		\end{align*}
		where $V$ is an isometry, which we use to define the map 
		\begin{align*}
			\tilde{\Theta}(\tau)=\Tr_C\left[ V^{BC\leftarrow \tilde{A}} \tau \left(V^{CB\leftarrow \tilde{A}}\right)^\dagger\right]
		\end{align*}
		that is CPTP by construction. Then
		\small
		\begin{align}\label{eq:OpOnReducedSpace}
			Y=&\Tr_C\left[\sum_{i,j}\sqrt{\sigma_i \sigma_j}\ketbra{i}{j}_A\otimes V^{BC\leftarrow \tilde{A}} \ketbra{i}{j}_{\tilde{A}} \left(V^{CB\leftarrow \tilde{A}}\right)^\dagger\right] \nonumber \\
			=&\sum_{i,j}\sqrt{\sigma_i \sigma_j} \ketbra{i}{j}_A\otimes\Tr_C\left[ V^{BC\leftarrow \tilde{A}} \ketbra{i}{j}_{\tilde{A}} \left(V^{CB\leftarrow \tilde{A}}\right)^\dagger\right] \nonumber \\
			=&\id^A\otimes \tilde{\Theta}^{B\leftarrow \tilde{A}} \sum_{i,j} \sqrt{\sigma_i \sigma_j}\ketbra{ii}{jj}_{A\tilde{A}}.
		\end{align}
		\normalsize
		Now, we introduce the map $\Pi$ defined by Kraus operators
		\begin{equation}\label{eq:ProjectionOp}
			K:=\sum_{i\in S}\ket{i}_{\tilde{A}}\bra{i}_A, \qquad L_j:=\ket{\psi}_{\tilde{A}}\bra{j}_A \quad\forall j\not\in S,
		\end{equation} where $\ket{\psi}_{\tilde{A}}$ is a normalized quantum state. \\
		It is easy to check that $K^\dagger K+\sum_{j\not\in S}L_j^\dagger L_j=\mathbbm{1}_A$ and that the map
		\begin{equation*}
			\Theta^{B\leftarrow A}:= \tilde{\Theta}^{B\leftarrow {\tilde{A}}}\Pi^{{\tilde{A}}\leftarrow A}
		\end{equation*} is such that 
		\begin{equation*}
			Y=\id^A\otimes {\Theta}^{B\leftarrow {A}} \sum_{i,j} \sqrt{\sigma_i \sigma_j}\ketbra{ii}{jj}_{AA}.
		\end{equation*}
		Since $\tilde{\sigma}=\sum_{i,j} \sqrt{\sigma_i \sigma_j}\ketbra{i}{j}=\ketbra{\phi}{\phi}:\ \ket{\phi}=\sum_i \sqrt{\sigma_i}\ket{i}$ is by assumption a valid quantum state, we showed that $Y$ is an element of $\mathcal{X}$. 
	\end{proof}
	
	\begin{lem}\label{lem:equiv}
		The two sets 
		\begin{align*}
			\mathcal{X}:= &\Bigg\{ X_{AB}=\id^A\otimes \Phi^{B\leftarrow A} \sum_{i,j} \rho_{i,j} \ketbra{ii}{jj}_{AA} \Bigg| \\
			&\pushright{ \Phi \in \mathcal{DI},\ \rho=\sum_{i,j}\rho_{i,j}\ketbra{i}{j} \text{ quantum state} \Bigg\},} \\
			\mathcal{Y}:= &\{Y_{AB} | Y\ge0,\ \Tr_BY=\Delta (\sigma_A),\ \\
			&\quad \diag\left(\bra{i}_AY\ket{j}_A\right)=0 \ \forall i\ne j,\ \sigma_A  \text{ quantum state} \}
		\end{align*}
		are equal.	
	\end{lem}
	\begin{proof}
		A quantum operation $\Phi$ is detection incoherent iff $\Delta \Phi=\Delta \Phi \Delta$.  Assume $X\in \mathcal{X}$. Using Lem.~\ref{lem:prelem1}, and the fact that
		\begin{align*}
			\diag\left(\bra{i}_AX\ket{j}_A\right)&=\diag\left(\Delta \Phi\left(\rho_{i,j}\ketbra{i}{j}\right)\right) \\
			&=\diag\left(\Delta \Phi \Delta \left(\rho_{i,j} \ketbra{i}{j}\right)\right),
		\end{align*}
		we find that $X\in \mathcal{Y}$.\\
		Now assume $Y\in \mathcal{Y}$. From the proof of Lem.~\ref{lem:prelem1}, we know that the assumptions ensure that we can write 
		\begin{align}\label{eq:ExpansionY}
			Y=&\id\otimes {\Phi} \sum_{i,j} \sqrt{\sigma_i \sigma_j}\ketbra{ii}{jj}
		\end{align}
		where ${\Phi}^{B\leftarrow A}$ is a quantum operation which is composed of $\tilde{\Phi}^{B\leftarrow {\tilde{A}}}\Pi^{{\tilde{A}}\leftarrow A}$. Then
		\begin{align*}
			\diag\left(\bra{i}_AY\ket{j}_A\right)=\sqrt{\sigma_i \sigma_j}\diag\left({\Phi}\ketbra{i}{j}\right)=0 \quad\forall i\ne j
		\end{align*}
		ensures that 
		\begin{align*}
			\Delta {\Phi}( \ketbra{i}{j})=\Delta {\Phi} \Delta (\ketbra{i}{j}) \quad \forall i\ne j .
		\end{align*}
		This allows us to conclude that ${\Phi}$ is detection incoherent. 
	\end{proof}
	
	Thanks to the previous Lemmas, we propose the following method to evaluate the pre-processed improvements numerically.
	\begin{theorem}\label{theo:semi_pre}
		Consider a quantum channel $\Theta^{C\leftarrow B}$ and let $N=\dim(C)$. Let further $(s_{m,n})_{m,n}$ be the matrix of dimension $2^N\times N$ that contains as rows all $N$-dimensional  vectors $\vec{s}_m$  whose entries are $\pm1$. The solution of the optimization problem
		\begin{equation}\label{eq:tosemi}
			F_{\lambda,\vec{\phi}}(\Theta^{C\leftarrow B} )= \max_{\Phi\in \mathcal{DI}}\left|\left|\Delta \Theta^{C\leftarrow B}  \Phi^{B\leftarrow A} \left(\lambda-\mu \Lambda_{\vec{\phi}}\right)\right|\right|_1
		\end{equation} is then equivalent to the maximum of the solutions of the following $2^N$ semidefinite programs
		\begin{align}\label{eq:semied}
			\begin{split}
				\text{maximize:}\quad\, &t_m \\
				\text{subject to:}\quad &t_m\leq \sum_{n=0}^{N-1} s_{m,n} \bra{n}_C \Theta^{C\leftarrow B} \left(Z\right) \ket{n}_C\\			
				&Z=\left(\sum_{i,j}\left(\lambda-\mu e^{i(\phi_i-\phi_j)}\right)\bra{i}_AX_{AB}\ket{j}_A\right) \\
				&X_{AB}\geq 0\\
				& \Tr_B(X_{AB})=\Delta(\sigma_A) \\
				& \sigma_A\geq 0 \\
				& \Tr(\sigma_A)=1 \\
				& \diag\left(\bra{i}_AX_{AB}\ket{j}_A\right)=0 \quad \forall i\ne j.
			\end{split}
		\end{align} 
	\end{theorem}
	
	\begin{proof}
		We begin by rewriting the optimization problem in Eq.~(\ref{eq:tosemi}) as
		\begin{align*}
			\begin{split}
				\text{maximize:} \quad\, &\Tr\left|\Delta \Theta \Phi \left(\lambda-\mu \Lambda_{\vec{\phi}}\right)(\rho)\right|\\
				\text{subject to:} \quad &\Phi\in \mathcal{DI} \\
				&\rho \geq0 \\
				&\Tr(\rho)=1.
			\end{split}
		\end{align*}
		Using the index representation from App.~\ref{ap:index}, the objective function can now be expanded as
		\small
		\begin{align}\label{eq:proofStepLemma}
			\Tr&\left|\Delta \Theta \Phi \left(\lambda-\mu \Lambda_{\vec{\phi}}\right)(\rho)\right| \\
			&=\sum_{n=0}^{N-1}\left|\bra{n}\Theta\left(\sum_{i,j,k,l}\Phi_{k,l}^{i,j}\left(\lambda-\mu e^{i(\phi_i-\phi_j)}\right)\rho_{i,j}\ketbra{k}{l}\right)\ket{n}\right| \nonumber \\
			&=\sum_{n=0}^{N-1}\Bigg|\bra{n}_B\Theta\Bigg(\sum_{i,j}\left(\lambda-\mu e^{i(\phi_i-\phi_j)}\right) \nonumber \\
			&\pushright{\bra{i}_A\left(\sum_{o,p}\rho_{o,p}\ketbra{o}{p}_A\otimes\sum_{k,l}\Phi_{k,l}^{o,p}\ketbra{k}{l}_B\right)\ket{j}_A\Bigg)\ket{n}_B\Bigg| \nonumber} \\
			&=\sum_{n=0}^{N-1}\Bigg|\bra{n}_B\Theta\Bigg(\sum_{i,j}\left(\lambda-\mu e^{i(\phi_i-\phi_j)}\right) \nonumber \\
			&\pushright{\bra{i}_A\id^A\otimes\Phi^{B\leftarrow A}\left(\sum_{o,p}\rho_{o,p}\ketbra{oo}{pp}_{AA}\right)\ket{j}_A\Bigg)\ket{n}_B\Bigg|.} \nonumber
		\end{align}
		\normalsize
		Hence, we reformulate our problem into
		\small
		\begin{align*}
			\begin{split}
				\text{maximize:} \quad\, &\sum_{n=0}^{N-1}\Bigg|\bra{n}_B\Theta\Bigg(\sum_{i,j}\left(\lambda-\mu e^{i(\phi_i-\phi_j)}\right)\bra{i}_A \\
				&\pushright{\id^A\otimes\Phi^{B\leftarrow A}\left(\sum_{o,p}\rho_{o,p}\ketbra{oo}{pp}_{AA}\right)\ket{j}_A\Bigg)\ket{n}_B\Bigg|}\\
				\text{subject to:} \quad &\Phi\in \mathcal{DI} \\
				&\rho \geq0 \\
				&\Tr(\rho)=1.
			\end{split}
		\end{align*}
		\normalsize
		Finally, using Lem.~\ref{lem:equiv} and 
		\begin{equation*}
			\sum_n|f_n|=\max_{\vec{s}_m} (\vec{s}_m\cdot\vec{f}),
		\end{equation*}
		where the vectors $\vec{s}_m$ have been introduced in the statement of the Theorem, we proved that Eq.~(\ref{eq:tosemi}) and the greatest value of the solutions of Eq.~(\ref{eq:semied}) are equivalent.
	\end{proof}
	
	As one can see, the evaluation method that we propose in the above Theorem unfortunately requires us to solve a number 
	of semidefinite programs that grows exponentially in the output dimension of $\Theta$. 
	However, the solution gives direct access to an optimal pair $\Phi_{\text{opt}}, \rho_{\text{opt}}$ of input state and pre-processing. This is a consequence of the constructive proofs of Lem.~\ref{lem:prelem1} and Lem.~\ref{lem:equiv} on which the method relies. Let us assume that we solved the semidefinite program in Eq.~\eqref{eq:semied} that leads to the maximal $t_m$ and denote its optimal $X$ with $X_{\text{opt}}$, from which one obtains $\sigma_{\text{opt},i}$ via
	\begin{align*}
		\Tr_B(X_{\text{opt}})=\sum_i \sigma_{\text{opt},i} \ketbra{i}{i}.
	\end{align*}
	
	Combining Eq.~\eqref{eq:ExpansionY} with Eq.~\eqref{eq:proofStepLemma},  an optimal $\rho$ is thus given by 
	\begin{align*}
		\rho_{\text{opt}}=\sum_{i,j}\sqrt{\sigma_{\text{opt},i} \sigma_{\text{opt},j}}\ketbra{i}{j}.
	\end{align*}
	With $\tilde{A}=\vspan\left\{\ket{i}: \sigma_{\text{opt},i}\ne0 \right\}$, we now define an operation $\tilde{\Phi}^{B\leftarrow\tilde{A}}_\text{opt}$ via Eq.~\eqref{eq:OpOnReducedSpace}, i.e., 
	\begin{align*}
		\tilde{\Phi}^{B\leftarrow\tilde{A}}_\text{opt}\left(\ketbra{i}{j}_{\tilde{A}}\right)=\frac{\bra{i}_A X_{\text{opt}}\ket{j}_A}{\sqrt{\sigma_{\text{opt},i} \sigma_{\text{opt},j}}}.
	\end{align*}
	Note that due to the definition of $\tilde{A}$, division by zero is excluded, and $\tilde{\Phi}^{B\leftarrow\tilde{A}}_\text{opt}$ is determined uniquely.
	Together with the operation $\Pi^{\tilde{A}\leftarrow A}$ introduced in Eq.~\eqref{eq:ProjectionOp}, which is also well defined due to the knowledge of $\sigma_{\text{opt},i}$, an optimal $\Phi^{B\leftarrow A}$ corresponding to the $\rho_{\text{opt}}$ given above is therefore 
	\begin{align}
		\Phi^{B\leftarrow A}_\text{opt}=\tilde{\Phi}^{B\leftarrow\tilde{A}}_\text{opt} \Pi^{\tilde{A}\leftarrow A}.
	\end{align}

	\section{Interpretation of faithfulness}\label{ap:peek}
	In this Appendix, we discuss the intuition behind the results of Thm.~\ref{thm:faithfulnessPre} concerning the faithfulness of the pre-processed improvements. 
	As example, we consider a detecting qubit operation $\Theta$ that is a stochastic mixture of the Hadamard gate $H$ and the identity channel, i.e.,
	\begin{equation}
		\Theta(\rho)=p_1 H\rho H^\dagger+p_2\rho.\label{eq:had-prob}
	\end{equation}
	For
	\begin{align}
		\vec{\phi}=\left(\frac{2\pi}{3},0\right),
	\end{align}
	in Fig.~\ref{fig:poslammmu} and Fig.~\ref{fig:neglammmu}, we plotted  $M_{\lambda,\vec{\phi}}(\Theta)$ with the help of Thm.~\ref{theo:semi_pre}.
	These plots clearly show that $M_{\lambda,\vec{\phi}}(\Theta)$ is not faithful for $\lambda\neq\mu$. 
	\begin{figure}[ht]
		\includegraphics[scale=0.75]{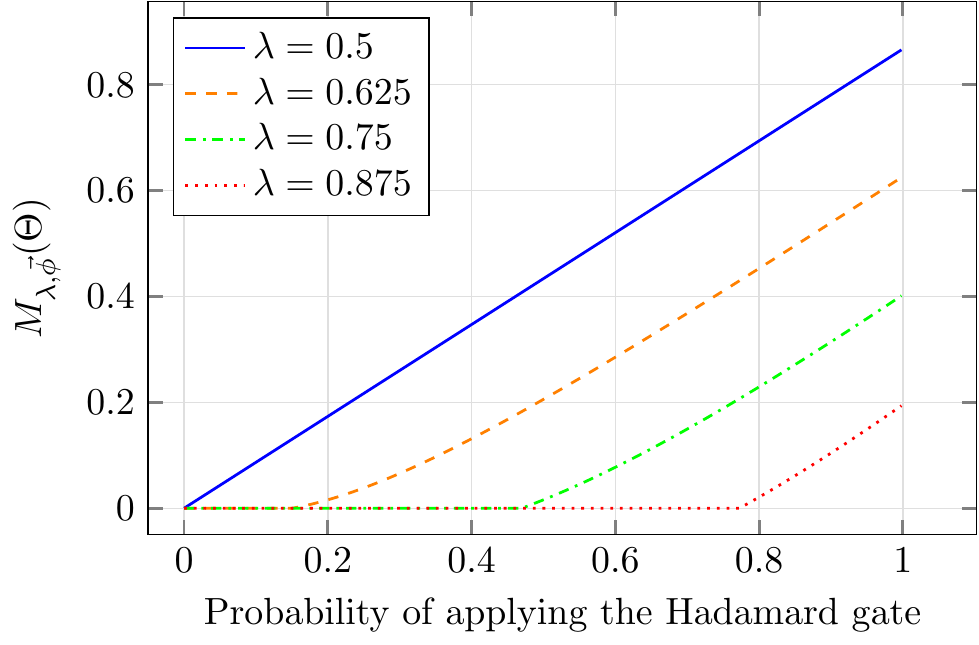}
		\caption{With $\vec{\phi}$ and $\Theta(p_1)$ as in the main text, where $p_1$ denotes the probability of applying the Hadamard gate, $M_{\lambda,\vec{\phi}}(\Theta)$ is plotted for different choices of $\lambda\geq\mu$. One clearly sees that $M_{\lambda,\vec{\phi}}(\Theta)$ is not faithful for $\lambda\ne\frac{1}{2}$. }
		\label{fig:poslammmu}
	\end{figure}
	\begin{figure}[ht]
		\includegraphics[scale=0.75]{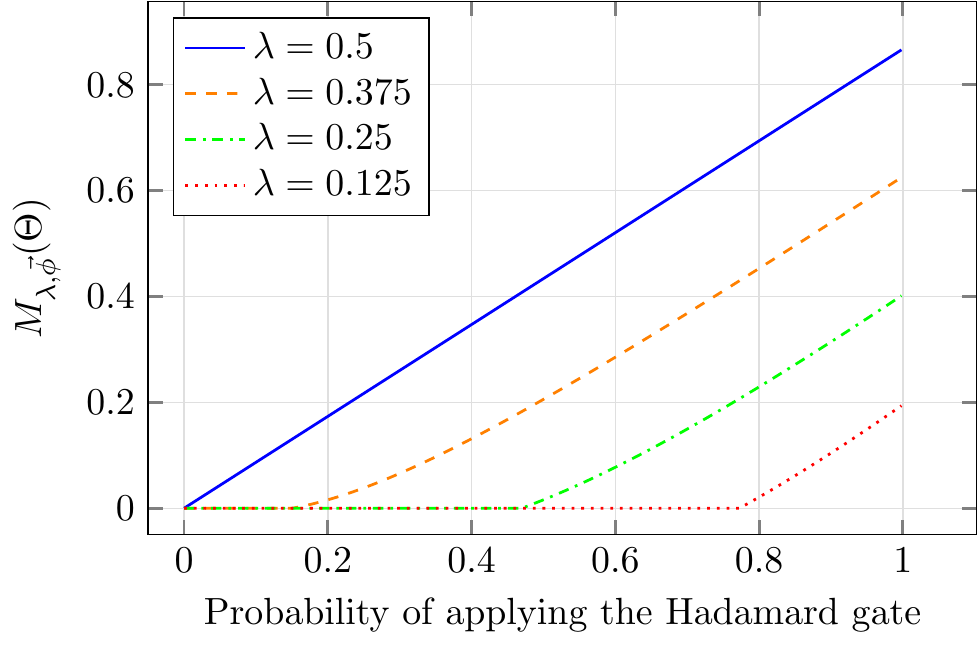}
		\caption{Plot corresponding to Fig.~\ref{fig:poslammmu} for $\lambda\leq\mu$.}
		\label{fig:neglammmu}
	\end{figure}
	This implies that, for $\lambda\ne\mu$,  not every operation able to detect coherence allows Bob to increase his probability of guessing correctly if Alice applied $\Lambda_{\vec{\phi}}$ or not.
	
	In the following, we will analyze the reasons for this fact. 
	Firstly, we notice that Bob can always guess correctly with a minimum probability of $\frac{1}{2}$ by not taking into account any information about Alice's actions but just announcing his guesses randomly with equal probability. 
	
	Secondly, by purely taking into account his knowledge of $\lambda$, he can increase this probability by the \emph{trivial bias} 
	\begin{equation}
		B_t:= \frac{1}{2}|\lambda-\mu |.
	\end{equation}
	To further increase his probability of guessing correctly, he needs to make use of information that he gains from the incoherent POVM.
	
	Now suppose that we want to distinguish two states $\sigma_0$ and $\sigma_1$ occurring with probabilities $\lambda$ and $\mu$ via an incoherent measurement. In this case, the maximal bias over $\frac{1}{2}$ that we can obtain is~\cite[Prop.~16]{Theurer2019}
	\begin{equation}
		B_m:= \frac{1}{2}||\Delta(\lambda\sigma_0-\mu\sigma_1)||_1,
	\end{equation}
	which we will call the \emph{measurement bias}.
	
	Intuitively, one might expect that $B_m$ is greater than $B_t$ for all pairs of states $\sigma_0$ and $\sigma_1$ that  have different populations, because then one can find an incoherent POVM that leads to different statistics for the two states. However, one deduces from Fig.~\ref{fig:poslammmu} and Fig.~\ref{fig:neglammmu} that there exist qubit states $\sigma_0=\Theta \Phi(\rho)$ and $\sigma_1=\Theta \Phi  \Lambda_{\vec{\phi}}(\rho)$ (with different populations for our costly $\Theta$ and ideal $\rho$ and $\Phi$), for which  $B_m=B_t$. 
	In addition, we note that for $\lambda\neq\mu$ these two plots show a discontinuity in the gradient of $M_{\lambda,\vec{\phi}}(\Theta)$ with respect to $p_1$.
	
	Since both $\sigma_0$ and $\sigma_1$ are qubit states, the respective measurement bias is  given by
	\begin{align}\label{eq:best_prot}
		B_m&= \frac{1}{2}\Tr|\Delta(\lambda\sigma_0-\mu\sigma_1)| \nonumber\\
		&=  \frac{1}{2}\Tr\left| (\lambda\sigma_{00}^{(0)}-\mu\sigma_{00}^{(1)})\ketbra{0}{0}+(\lambda\sigma_{11}^{(0)}-\mu\sigma_{11}^{(1)})\ketbra{1}{1}\right|   \nonumber\\
		&= \frac{1}{2}\max\left\{ \left|\Tr (\lambda\sigma_0-\mu\sigma_1)\right| , \left|\Tr(\sigma_z(\lambda\sigma_0-\mu\sigma_1))\right| \right\}  \nonumber \\
		&= \frac{1}{2}\max\left\{ \left|\lambda-\mu \right| , \left|\Tr(\sigma_z(\lambda\sigma_0-\mu\sigma_1))\right| \right\}. 
	\end{align} 
	From Ref.~\cite[Prop.~16]{Theurer2019} and its proof, we know that an optimal guessing strategy for the distinction of $\sigma_0$ and $\sigma_1$ that involves only incoherent measurements   consists of the following: measure the POVM $\{P_0, P_1=\mathbbm{1} -P_0\}$, where $P_0$ is the projector onto the positive part of $\Delta(\lambda\sigma_0-\mu\sigma_1)$ and $P_1$ the projector onto its negative part, and announce $i=0,1$ according to the outcome.
	From the second line of Eq.~\eqref{eq:best_prot}, we thus find the following four cases for our qubit example
	\begin{itemize}
		\item If $\lambda\sigma_{00}^{(0)}-\mu\sigma_{00}^{(1)}$ and $\lambda\sigma_{11}^{(0)}-\mu\sigma_{11}^{(1)}$ are both non-negative, we have $P_0=\mathbbm{1}$. In other words, Bob does not need to do any measurement and always claims that the phases have not been attached.
		
		\item If neither $\lambda\sigma_{00}^{(0)}-\mu\sigma_{00}^{(1)}$ nor $\lambda\sigma_{11}^{(0)}-\mu\sigma_{11}^{(1)}$ are positive, we find $P_1=\mathbbm{1}$ and Bob always claims that the phases have been attached. Thus, in these first two cases, no free measurement leads to useful information.
		
		\item If $\lambda\sigma_{00}^{(0)}-\mu\sigma_{00}^{(1)}$ is positive and $\lambda\sigma_{11}^{(0)}-\mu\sigma_{11}^{(1)}$ negative, we choose $P_0=\ketbra{0}{0}$, i.e., actually perform a measurement. Here, according to the outcome, Bob should declare that Alice encoded $\vec{\phi}$ or that she did not. 
		
		\item In the remaining case where $\lambda\sigma_{00}^{(0)}-\mu\sigma_{00}^{(1)}$ is negative and $\lambda\sigma_{11}^{(0)}-\mu\sigma_{11}^{(1)}$ positive, we find $P_0=\ketbra{1}{1}$. In the last two cases, we can actually gain additional information via a free measurement.
	\end{itemize}
	This also explains the discontinuities in the gradients, which occur at the points from  which on performing an incoherent measurement leads to an actual advantage. The above four cases are collected in the third line of Eq.~\eqref{eq:best_prot} (where the POVMs are expressed by the measurement of the observable $\pm\sigma_z$).
	
	This behavior is actually not due to the incoherent measurements or any other quantum property, but can already be explained with the following purely classical example.
	Imagine that Alice has two marbles, a dark green one and a blue one. She chooses one of them with an a priori probability that is known to Bob and throws it into a dark room. Now Bob looks at it, but since the room is dark, blue and dark green are hard to distinguish, hence Bob is not sure about the color of the marble he sees. What is now his best strategy to guess correctly which one it is? Betting based on the knowledge of the a priori probability or based on the color he believes to see? Or combining both? Intuitively, he should make the bet based on the information that is most secure. Therefore, if he knows that Alice throws the blue marble with a probability of $95\%$, but according to his eyes (which we assume to be reliable with a probability of $60\%$  under such conditions) it is the green one, he should not trust his eyes. Independent of whether he looked at the marble or not, Bob's ideal guessing strategy is to always claim that the marble is blue. As the percentages change, at some point, the ideal guessing strategy starts to depend on what he sees.  Considering the specific case that the a priori probability for the marbles is $50\%$, this will always be the case (assuming that his eyes gather any useful information), which is the reason why $M_{\lambda,\vec{\phi}}(\Theta)$ is faithful for $\lambda=\frac{1}{2}$.

	%

\end{document}